\documentclass[onecolumn,11pt,draftcls]{IEEEtran}
\usepackage{amsbsy}%,color,multicol}
\usepackage{floatflt} % text flows around figures

\usepackage{amsmath}
\usepackage{amssymb}
\usepackage{times}
\usepackage{graphicx}
\usepackage{xspace}
\usepackage{paralist} % compact and in-paragraph* lists
\usepackage{setspace} % buggy somehow
\usepackage{xypic}
\xyoption{curve}
\usepackage{latexsym}
\usepackage{theorem}
\usepackage{ifthen}
\usepackage{subfigure}
\usepackage{booktabs}

\newcommand{\mypar}[1]{\vspace{0.03in}\noindent{\bf #1.}}

{\theoremheaderfont{\it} \theorembodyfont{\rmfamily}
\newtheorem{theorem}{Theorem}
\newtheorem{lemma}[theorem]{Lemma}
\newtheorem{definition}[theorem]{Definition}

\newtheorem{assumption}[theorem]{Assumption}

}

\begin{document}
\title{ Distributed Detection over Random Networks: Large Deviations Performance Analysis}

\author{Dragana Bajovi\'c, Du$\breve{\mbox{s}}$an Jakoveti\'c, Jo\~ao
Xavier, Bruno Sinopoli and Jos\'e M.~F.~Moura % <-this
\thanks{Partially supported by grants SIPM PTDC/EEA-ACR/73749/2006 and
SFRH/BD/33517/2008 (through the Carnegie Mellon/Portugal Program
managed by ICTI) from Funda\c{c}\~{a}o para a Ci\^encia e Tecnologia and
also by
ISR/IST plurianual funding (POSC program, FEDER). Work of Jos\'e~M.~F.~Moura is partially supported by NSF under grants CCF-1011903 and CCF-1018509, and by AFOSR grant
FA95501010291. Dragana Bajovi\'{c} and Du$\breve{\mbox{s}}$an Jakoveti\'c hold fellowships
from the Carnegie Mellon/Portugal Program.}% <-this % stops a space
\thanks{Dragana Bajovi\'{c} and Du$\breve{\mbox{s}}$an Jakoveti\'c are with the
Institute for Systems and Robotics
(ISR), Instituto Superior T\'{e}cnico (IST), Lisbon, Portugal, and with
the Department of Electrical and Computer Engineering, Carnegie Mellon
University, Pittsburgh, PA, USA {\tt\small dbajovic@andrew.cmu.edu,
djakovet@andrew.cmu.edu}}%
\thanks{Jo\~ao Xavier is with the Institute for Systems and Robotics (ISR),
Instituto Superior T\'{e}cnico (IST), Lisbon, Portugal {\tt\small
jxavier@isr.ist.utl.pt}}
\thanks{Bruno Sinopoli and Jos\'e M.~F.~Moura are with the Department of
Electrical and Computer
Engineering, Carnegie Mellon University, Pittsburgh, PA, USA {\tt\small
brunos@ece.cmu.edu, moura@ece.cmu.edu; ph: (412)268-6341; fax: (412)268-3890}}}%
%
%

%\author{Dragana Bajovi\'c, Du$\breve{\mbox{s}}$an Jakoveti\'c, Jo\~ao
%Xavier, Bruno Sinopoli and Jos\'e M.~F.~Moura}

\maketitle

\begin{abstract}
We study the large deviations performance, i.e., the exponential decay rate of the error probability, of distributed detection algorithms over random networks. At each time step $k$ each sensor: 1) averages its decision variable
with the neighbors' decision variables; and 2) accounts on-the-fly for its new observation. We show that
 distributed detection exhibits a ``phase change'' behavior. When the rate of network information flow (the speed of averaging)
  is above a threshold, then distributed detection is asymptotically equivalent to the optimal centralized detection, i.e.,
  the exponential decay rate of the error probability for distributed detection equals the Chernoff information.
  When the rate of information flow is below a threshold, distributed detection achieves only a fraction of the
  Chernoff information rate; we quantify this achievable rate as a function of the network rate of information flow. Simulation
  examples demonstrate our theoretical findings on the behavior of distributed detection over random networks.
\end{abstract}
\hspace{.43cm}\textbf{Keywords:} Chernoff information, distributed detection,
random network, running consensus, information flow, large deviations.
\newpage
\section{Introduction}
\label{section-intro}

Existing literature on distributed detection can be
broadly divided into three different classes.
The first studies parallel (fusion) architectures, where all sensors
transmit their measurements, or local likelihood ratios, or local
decisions, to a fusion node; the fusion node subsequently
makes the final decision (see, e.g.,~\cite{Varshney-I,Veraavali,Tsitsiklis-detection,Poor-II}.)
The second considers consensus-based detection, where no fusion node is
required, and sensors communicate with single-hop
neighbors only over a generic network (see, e.g.,~\cite{Moura-detection-consensus,moura-cons-detection}).
Consensus-based detection operates in two phases. First, in the sensing phase, each
sensor collects sufficient observations over a period of time. In the second, communication phase,
 sensors subsequently run the consensus algorithm to fuse their local log
likelihood ratios. More recently, a third class of distributed detection has been proposed
(see~\cite{running-consensus-detection,Sayed-detection-2,Sayed-detection},)
where, as with consensus-based detection,
  sensors communicate over a generic network, and no fusion node is
required. Differently than consensus-based
  detection, sensing and communication phases occur in the same time step.
Namely, at each time step $k$,
   each sensor: 1) exchanges its current decision variable with single-hop
neighbors; and 2) processes its new observation, gathered at time step $k$.

In this paper, we focus on the third class of distributed detection, and
we provide fundamental analysis of the large deviation performance,
i.e., of the exponential decay rate of the error probability (as $k \rightarrow \infty$,) when the underlying
communication network is \emph{randomly varying}. Namely, we show that distributed detection over random networks
is asymptotically equivalent to the optimal centralized detection,
if the rate of information flow (i.e., the speed of averaging) across the random network is large enough.
That is, if the rate of information flow is above a threshold, then the exponential
rate of decay of the error probability of distributed detection equals the Chernoff information--the best possible rate of the optimal centralized detector. When the random network has slower information flow (asymptotic optimality cannot be achieved,)
     we find what fraction of the best possible rate of decay of the error probability distributed detection can achieve; hence, we quantify the tradeoff between the network connectivity and achievable detection performance.

Specifically, we consider the problem where sensors
cooperate over a network and sense the environment to decide between two hypothesis.
The network is random, varying over time
$k$ (see, e.g., \cite{weight-opt}); in alternative, the network uses a random communication protocol, like gossip (see, e.g.,~\cite{BoydGossip}).
The network connectivity is described by $\{W(k)\}_{k=1}^\infty$, the sequence of identically distributed (i.i.d.)
consensus weight matrices. The sensors' observations are
Gaussian, correlated in space, and uncorrelated in time. At each time $k$, each sensor: 1)
 communicates with its single-hop neighbors to compute the weighted average of its own and the
 neighbors' decision variables; and 2) accounts for its new observation acquired at time $k$.
  The network's rate of information flow (i.e., the speed of averaging,) is then measured by $|\log r|$, where $r$ is the
   second largest eigenvalue of the expected value of $W(k)^2$. We then show that
    distributed detection exhibits a ``phase change'' behavior. If $|\log r|$ exceeds
   a threshold, then distributed detection is asymptotically optimal.
    If $|\log r|$ is below the threshold, then distributed detection
     achieves only a fraction of the best possible rate; and we evaluate the achievable rate as
      a function of $r$. Finally, we demonstrate by simulation examples our
      theoretical findings on the behavior of distributed detection over random networks.

Several recent references~\cite{Sayed-detection-2,Sayed-detection,running-consensus-detection} consider different variants of distributed detection of the third class. We consider in this paper the running consensus, the variant in~\cite{running-consensus-detection}.

 In the context of estimation, distributed iterative schemes have also been proposed. References~\cite{Sayed-LMS,Sayed-LMS-new}
   propose diffusion type LMS and RLS algorithms
    for distributed estimation; references~\cite{Giannakis-LMS,Giannakis-LMS-2}
    also propose algorithms for distributed estimation, based on the alternating direction method of multipliers.
     Finally, reference~\cite{SoummyaEst} proposes stochastic-approximation
      type algorithm for distributed estimation, allowing for
       randomly varying networks and generic (with finite second moment) observation noise.
       With respect to the network topology and the observation noise, we also
       allow for random networks, but we assume Gaussian, spatially correlated observation noise.

We comment on the differences between this work and reference~\cite{running-consensus-detection}, which also studies asymptotic performance of distributed detection via running consensus, with i.i.d. matrices $W(k)$. Reference~\cite{running-consensus-detection} studies a problem different than ours, in which the means of the sensors' observations under the two hypothesis become closer and closer; consequently, there is an asymptotic, non zero, probability of miss, and asymptotic, non zero, probability of false alarm. Within this framework, the running consensus achieves the efficacy~\cite{Kassam} of the optimal centralized detector, under a mild assumption on the underlying network being connected on average. In contrast, we assume that the means of the distributions do not approach each other as $k$ grows, but stay fixed with $k$. The Bayes error probability exponentially decays to zero, and we examine its rate of decay. We show that, in order to achieve the optimal decay rate of the Bayes error probability, the running consensus needs an assumption \emph{stronger} than connectedness on average, namely, the averaging speed needs to be \emph{sufficiently large} (as measured by $|\log r|$.)

In recent work~\cite{allerton}, we considered running consensus detection when the underlying network is \emph{deterministically} time varying; we showed that asymptotic optimality holds if the graph that collects the union of links that are online at least once over a finite time window is connected. In contrast, we consider here the case when the underlying network or the communication protocol are \emph{random}, and we establish
a sufficient condition for optimality in terms of averaging speed (measured by $|\log r|$.)

\mypar{Paper organization} The next paragraph defines notation that we use throughout
the paper. Section~{II} reviews standard asymptotic results in hypothesis testing, in particular,
 the Chernoff lemma. Section~{III} explains the sensor observations model that we assume and
   studies the optimal centralized detection, as if there was a fusion node
   to process all sensors' observations. Section~{IV}
    presents the running consensus distributed detection algorithm.
    Section~{V} studies the asymptotic performance of distributed detection on a simple, yet illustrative, example
      of random matrices $W(k)$. Section {VI} studies asymptotic performance of distributed detection in the general case. Section {VII}
       demonstrates by simulation examples our theoretical findings. Finally, section {VIII} concludes the paper.

\mypar{Notation} We denote by: $A_{ij}$ or $\left[ A\right]_{ij}$ (as appropriate) the $(i,j)$-th entry of a matrix $A$; $a_i$
 or $[a]_i$ the $i$-th entry of a vector $a$; $I$, $1$, and $e_i$, respectively, the identity matrix, the column vector with unit entries, and the
$i$-th column of $I$, $J$ the $N \times N$ matrix $J:=(1/N)11^\top$; $\| \cdot \|_l$ the vector (respectively, matrix) $l$-norm of its vector (respectively, matrix) argument, $\|\cdot\|=\|\cdot\|_2$ the Euclidean (respectively, spectral) norm of its vector
(respectively, matrix) argument, $\|\cdot\|_F$ the Frobenius norm of a matrix; $\lambda_i(\cdot)$ the $i$-th largest eigenvalue, $\mathrm{Diag}\left(a\right)$ the diagonal matrix with the diagonal equal to the vector $a$; $\mathbb E \left[ \cdot \right]$ and $\mathbb P \left( \cdot\right)$ the expected value and probability, respectively; $\mathcal{I}_{\mathcal{A}}$ the indicator function of the event $\mathcal{A}$; finally, $\mathcal{Q}(\cdot)$ the Q-function, i.e., the function that calculates the right tail probability
 of the standard normal distribution; $\mathcal{Q}(t)=\frac{1}{\sqrt{2 \pi}} \int_{t}^{+\infty} e^{-\frac{u^2}{2}}du$, $t \in \mathbb R$.

%
%
%We denote matrices by capital letters, e.g., $A$,
%and vectors and scalars by lower case letters, e.g., $a$.
%Entries of vector $a$ (respectively, matrix $A$)
%are denoted by $a_i$ or $[a]_i$ (respectively $A_{ij}$ or $[A]_{ij}$.)
%Symbols $\| a \|$ and $\rho(A)$
%denote the Euclidean norm of vector $a$, and 2-norm
% of matrix $A$, respectively. Symbols $I$ and $e_i$
%  denote the identity matrix and the $i$th column of identity
%   matrix, respectively. Symbol $\mathrm{Diag}(a)$ denotes the diagonal matrix with
% its diagonal equal to vector $a$; $\mathrm{vec}(A)$
% is a vector that stacks columns of $A$. Symbol $A \otimes B$ denotes the Kronecker product of the
%matrices $A$ and $B$. Specifically, we
% will, if suitable, denote $A \otimes A = A^{2,K}$; the $m$-th Kronecker power-product of
% the matrix $A$ is denoted by $A^{m,K}$. Symbols $\mathbb E \left[ a \right]$ and
% $\mathbb P \left( \mathcal{A}\right)$  denote the expected value of a random variable
% $a$ and the probability of an event $\mathcal{A}$, respectively.
%
\vspace{-3mm}
\section{Preliminaries}
\label{section-preliminaries}
This section reviews standard asymptotic results in hypothesis
 testing, in particular, the Chernoff lemma,~\cite{Cover}; it also
  introduces certain inequalities for the $\mathcal{Q}(\cdot)$ function that we use throughout.  We first formally define the binary hypothesis testing problem
  and the log-likelihood ratio (LLR) test.

\mypar{Binary hypothesis testing problem: Log-likelihood ratio test} Consider the sequence of independent identically distributed (i.i.d.)
$d$-dimensional random vectors (observations) $y(k)$, $k=1,2,...$, and the
binary hypothesis testing problem of deciding whether
 the probability measure generating $y(k)$ is $\nu_0$ (under
hypothesis $H_0$) or $\nu_1$ (under $H_1$).
  Assume that $\nu_1$ and $\nu_0$ are mutually
 absolutely continuous, distinguishable measures. Based on the
 observations $y(1),...,y(k)$, formally, a decision
 test $T$ is a sequence of maps $T_k: {\mathbb R}^{kd} \rightarrow
\{0,1\}$, $k=1,2,...$, with the interpretation that
$T_k(y(1),...,y(k))=l$ means that $H_l$ is decided, $l=0,1$.
Specifically, consider the log-likelihood ratio (LLR) test to decide
between $H_0$ and $H_1$, where the $T_k$ is given as follows:
  \begin{eqnarray}
  \label{eqn-llr-test-generic}
  \mathcal{D}(k)&:=&\frac{1}{k} \sum_{j=1}^k L(k) \\
  T_k &=& \mathcal{I}_{\{\mathcal{D}(k)>\gamma_k\}},
  \end{eqnarray}
 where $L(k):=\log \frac{d \nu_1}{d \nu_0}\left(y(k)\right)$ is the LLR
(given by the Radon-Nikodym derivative of
 $\nu_1$ with respect to $\nu_0$ evaluated at $y(k)$,) and $\gamma_k$ is a
chosen threshold.

%\mypar{LLR test: Large deviations} It can be shown that the sequence of LLR's $\{L(k)\}$, conditioned on
%$H_l$, $l=0,1$, is i.i.d. Denote
%by $\mu_k^{(l)}$ the probability measure of $\mathcal{D}(k)$ under
%hypothesis $H_l$. Applying
%Cram\'er's Theorem (Theorem~\ref{theorem-cramer}), it can be shown that the
%sequence of measures $\{\mu_k^{(l)}\}$, $l=0,1$,
% satisfies the LDP with good\footnote{Goodness of rate function is
%compactness of its sublevel sets.} rate function:
% %
% %
% \begin{equation}
% \label{eqn-lambda-generic}
% \Lambda^\star_{(l)}(t) = \sup_{\lambda \in \mathbb R}  \left( \lambda t -
%\Lambda_{(l)}(\lambda) \right),
% \end{equation}
% %
%where $\Lambda_{(l)}(\cdot)$ is the log-moment generating function of
%$L(k)$ under hypothesis $H_l$:
%\begin{equation}
%\Lambda_{(l)}(\lambda) = \log \mathbb E \left[ e^{\lambda L(k)}| H_l \right].
%\end{equation}
%%
%%That is, the rate function $\Lambda_{(l)}^\star(t)$ is the
%%Fenchel-Legendre (F-L) (\cite{DemboZeitouni}) transform of the log-moment
%%generating function of $L(k)$ under $H_l$.
%It can be shown that
%$\Lambda_{(1)}^\star(t) = \Lambda_{(0)}^\star(t)-t$. We summarize this
%result in the following theorem, e.g.,~\cite{DemboZeitouni}:
%%
%\begin{theorem}
%\label{theorem-ldp}
%The sequence of measures $\{ \mu_k^{(l)} \}$ of $\mathcal{D}(k)$ under
%$H_l$ satisfies the LDP with good rate function given by
%eqn.~\eqref{eqn-lambda-generic}.
%\end{theorem}
%%
%
\mypar{Asymptotic Bayes detection performance: Chernoff lemma}
%\label{subsect-chernoff-lemma}
%
%
Given a test $T$, we are interested in quantifying the detection performance,
 namely, in determining the Bayes error probability after $k$ data (observation) samples are processed:
\begin{equation}
 \label{eqn-bayes-P-e}
P^e(k) = P \left( H_0 \right) \alpha(k) + P \left( H_1 \right) \beta(k),
\end{equation}
where $P \left( H_l \right)$ are the prior probabilities, $\alpha(k)=\mathbb{P}\left( T_k=1|H_0\right)$ and
$\beta(k)=\mathbb{P} \left( T_k=0|H_1\right)$ are, respectively,
the probability of false alarm and the probability of a miss. Generally,
exact evaluation of $\alpha(k)$ and $\beta(k)$ (and hence, $P^e(k)$)
is very hard (as in the case of distributed detection over random networks
 that we study; see also~\cite{Moura-saddle-point} for distributed detection on a parallel architecture.) We seek computationally tractable
estimates of $P^e(k)$, when $k$ grows large. Typically, for large $k$, $P^e(k)$
is a small number (i.e., the detection
   error occurs rarely,) and, in many models, it exponentially decays to zero as $k \rightarrow
+\infty$. Thus, it is of interest
 to determine the (large deviations) rate of exponential decay of $P^e(k)$, given by:
 \begin{equation}
 \label{eqn-lim-p-e}
 \lim_{k \rightarrow \infty} \,-\frac{1}{k} \log P^e(k).
 \end{equation}
Lemma~\ref{chernoff-lemma} (\cite{Cover,DemboZeitouni}) states that, among all
possible decision tests, the LLR test with zero threshold maximizes~\eqref{eqn-lim-p-e}
 (i.e., has the fastest decay rate of $P^e(k)$.)
The corresponding decay rate equals the Chernoff information ${\bf{C}}$, i.e., the Chernoff distance between
the distributions of $y(k)$ under $H_0$ and $H_1$, where ${\bf{C}}$ is given by,~\cite{Cover}:
\begin{equation}
\label{eqn-cher-infor}
{\bf{C}} = \max_{s \in [0,1]} \left\{- \log \int \left(\frac{d \nu_0}{d \nu_1}\right)^{1-s} d \nu_1 \right\}.
\end{equation}
\begin{lemma}[Chernoff lemma]
\label{chernoff-lemma}
If $P(H_0) \in (0,1)$, then:
\begin{equation}
\label{eqn-theorem-chernoff-lemma}
\sup_{T} \left\{\limsup_{k \rightarrow \infty} \,-\frac{1}{k} \log P^e(k)
\right\} = {\bf{C}},
\end{equation}
where the supremum over all possible tests $T$ is attained for the LLR test
with $\gamma_k=0$, $\forall k$.
\end{lemma}
%
%
%The quantity $\Lambda_{(0)}^\star(0)=\Lambda_{(1)}^\star(0)$ is called the
%Chernoff distance between
%the distributions of $y(k)$ under $H_0$ and $H_1$, or Chernoff
%information,~\cite{DemboZeitouni}.\footnote{In information theoretic and
% signal processing literature, a more common expression for the Chernoff information is
%  $\max_{s \in [0,1]} \left\{- \log \int \left(\frac{d \nu_0}{d \nu_1}\right)^{1-s} d \nu_1 \right\}$, see~\cite{Cover}. It can be shown that latter expression equals $\Lambda^\star_{(0)}(0).$}
%
%
%

\mypar{Asymptotically optimal test} We introduce the following definition
of the asymptotically optimal test.
\begin{definition}
\label{def-asym-optimal}
The decision test $T$ is asymptotically optimal if it attains the supremum
in eqn.~\eqref{eqn-theorem-chernoff-lemma}.
\end{definition}
We will find a necessary condition and a sufficient condition for
asymptotic optimality (in the sense of
Definition~\ref{def-asym-optimal}) of the running consensus distributed detection.

\mypar{Inequalities for  the standard normal distribution}
We will use the following property of the $\mathcal{Q}(\cdot)$ function, namely, that for any $t>0$ (e.g., \cite{Harry-van-Trees-book}):
\begin{equation}
\label{eqn-Q-ineq}
\frac{t}{1+t^2} e^{-\frac{t^2}{2}} \leq \sqrt{2 \pi}\, \mathcal{Q}(t)  \leq \frac{1}{t} e^{-\frac{t^2}{2}} .
\end{equation}

\section{Centralized detection}
We proceed with the Gaussian model for which we find (in section~V) conditions for asymptotic optimality of
the running consensus distributed detection. Subsection~\ref{subsect-data-model} describes the model of the sensor observations that we assume. Subsection~{III-B}
  describes the (asymptotically) optimal centralized detection, as if there was a fusion node that collects and processes the observations from all sensors.
\subsection{Sensor observations model}
\label{subsect-data-model}
We assume that $N$ sensors are deployed to sense
the environment and to decide
between the two possible hypothesis, $H_1$ and $H_0$. Each sensor $i$
  measures a scalar quantity $y_i(k)$ at each time step $k$; all sensors
  measure at time steps $k=1,2,...$ Collect $y_i(k)$'s,
   $i=1,...,N$, into $N \times 1$ vector $y(k)=(y_1(k),...,y_N(k))^\top$.
    We assume that $\{y(k)\}$ has the following distribution:
     \begin{equation}
     \mathrm{Under\,\,}H_l:\,y(k) =m_l + \zeta(k), \,\,l=0,1.
     \end{equation}
The quantity $m_l$ is the constant signal; the quantity $\zeta(k)$ is zero-mean, Gaussian, spatially correlated noise,
   i.i.d. across time, with distribution $\zeta(k) \sim \mathcal{N}\left(0, S \right)$,
    where $S$ is a positive definite covariance matrix.
    Spatial correlation of the measurements (i.e., non-diagonal covariance matrix $S$) accounts for, e.g., dense deployment
     in sensor networks.
 %is an independent identically distributed (i.i.d.) sequence of $N \times 1$
% random vectors with distribution $\zeta(k) \sim \mathcal{N}(0,S)$, where $S$
%  is a (positive definite) covariance matrix. Thus, with our model,
%  the noise is temporally independent, but can be spatially correlated.
%  Spatial correlation should be taken into account due to, for example,
%  dense deployment of wireless sensor networks, while it is still
%reasonable to assume
%  that the observations are independent along time. (Conditioned to
%   $H_l$, $\{y(k)\}$ are i.i.d. with
%   the distribution $\mathcal{N}(m_l,S)$.)

%
\subsection{(Asymptotically) optimal centralized detection}
This subsection studies optimal centralized detection under the Gaussian assumptions in~\ref{subsect-data-model},
 as if there was a fusion node that collects and processes all sensors' observations. The LLR decision test is given by eqns.~(1)~and~(2), where it is straightforward to show that now the LLR takes the following form:
\begin{equation}
\label{eqn_LLR}
L(k) = (m_1-m_0)^\top S^{-1} \left(   y(k) - \frac{m_1+m_0}{2} \right)
\end{equation}
Conditioned on either hypothesis $H_1$ and $H_0$,
$L(k)\sim\mathcal{N} \left( m_L^{(l)}, \sigma_L^2 \right)$, where
\begin{eqnarray}
\label{eqn-m_L-sigma-L}
m_L^{(1)} &=&  -m_L^{(0)} = \frac{1}{2} (m_1-m_0)^\top S^{-1} (m_1-m_0)\\
\sigma_L^2 &=& (m_1-m_0)^\top S^{-1} (m_1-m_0).
\end{eqnarray}
Define the vector $v \in {\mathbb R}^N$ as
\begin{equation}
\label{eqn-def-v}
v:=S^{-1}(m_1-m_0).
\end{equation}
Then, the LLR $L(k)$ can be written as follows:
\begin{equation}
\label{eqn-L(k)-summable}
L(k) = \sum_{i=1}^N v_i \left( y_i(k) - \frac{[m_1]_i+[m_0]_i}{2} \right)
= \sum_{i=1}^N \eta_i(k)
\end{equation}
Thus, the LLR at time $k$ is separable, i.e., the LLR
is the sum of the terms $\eta_i(k)$ that depend affinely on the individual
observations $y_i(k)$.
 We will exploit this fact in section~{IV} to derive
the distributed, running consensus, detection
 algorithm.

\mypar{Bayes probability of error: finite number of observations} The minimal Bayes error probability, $P^e_{\mathrm{cen}}(k)$, when $k$ samples are processed, and $P(H_0)=P(H_1)=\frac{1}{2}$ (equal prior probabilities), is attained for the (centralized) LLR test with
zero threshold; $P^e_{\mathrm{cen}}(k)$ equals:
 \begin{equation}
 \label{eqn-P-e-centr}
 P^e_{\mathrm{cen}}(k) = \mathcal{Q} \left( \sqrt{k}\frac{m_L^{(1)}}{\sigma_L} \right) .
 \end{equation}
 The quantity $P^e_{\mathrm{cen}}(k)$ will be of interest when we
 compare (by simulation, in Section~{VII}) the running consensus detection with the optimal centralized detection, in the
 regime of finite $k$.

\mypar{Bayes probability of error: time asymptotic results}
The Chernoff lemma (Lemma~\ref{chernoff-lemma}) applies also to the
(centralized) detection problem as defined in subsection {III}-A. It can be shown
 that the Chernoff information, in this case, equals:
 \begin{equation}
 \label{eqn-C-inf-gauss}
 {\bf{C}} = {\bf{C}_{\mathrm{tot}}} = \frac{1}{8} (m_1-m_0)^\top S^{-1} (m_1-m_0).
 \end{equation}
In eqn.~\eqref{eqn-C-inf-gauss}, the subscript $\mathrm{tot}$
 designates the total Chernoff information of the network, i.e.,
  the Chernoff information of the observations collected from all sensors.
   Specifically, if the sensor
   observations are uncorrelated (the noise covariance matrix $S=\mathrm{Diag}(\sigma_1^2,...,\sigma_N^2)$,)
   then:
   \begin{equation}
   \label{C-tot-uncorr}
   {\bf{C}_{\mathrm{tot}}} = \sum_{i=1}^N \frac{[m_1-m_0]_i^2}{8 \sigma_i^2}=\sum_{i=1}^N {\bf{C}_{i}},
   \end{equation}
   where ${\bf{C}}_i$ is the Chernoff information of
   the individual sensor $i$. That is,
   ${\bf{C}}_i$ equals the best achievable
   rate of the Bayes error probability,
   if the sensor $i$ worked as an individual (it did not cooperate with the other sensors.)

%, for the Gaussian (centralized) detection as defined in~\ref{subsect-data-model}, simplifies to the following corollary:
  %
  %
  %
  \begin{lemma}(\emph{Chernoff lemma for asymptotically optimal centralized detector})
  \label{lemma-Chernoff-gauss}
  Consider the observation model defined in subsection~{III-A},
   and let $P(H_0) \in (0,1)$. The LLR test with $\gamma_k=0$, $\forall k$, is asymptotically optimal
in the sense of Definition~\ref{def-asym-optimal}. Moreover, for the LLR test with
$\gamma_k=0$, $\forall k$, we have:

  \begin{equation}
  \lim_{k \rightarrow \infty} \frac{1}{k} \log P^e_{\mathrm{cen}}(k) = - { \bf{C}_{\mathrm{tot} }} ,
  \end{equation}
where ${\bf{C}_{\mathrm{tot}}}$ is given by eqn.~\eqref{eqn-C-inf-gauss}.
  \end{lemma}
\section{Distributed detection}
We now consider distributed detection, under the same assumptions on the sensor observations as in~\ref{subsect-data-model}; but
the fusion node is no longer available, and the sensors cooperate through a randomly varying network. Specifically,
  we consider the running consensus distributed detection, proposed in~\cite{running-consensus-detection}, and we extend
  it to spatially correlated observations.
   At each time $k$, each sensor $i$ improves its decision variable,
    call it $x_i(k)$, two-fold: 1) by exchanging the decision variable
     locally with its neighbors and computing the weighted average
     of its own and the neighbors' variables; and 2)
      by incorporating its new observation at time $k$.

      Recall the definition of the vector $v$ in eqn.~\eqref{eqn-def-v} and the scalar
      $\eta_i(k)$ in eqn.~\eqref{eqn-L(k)-summable}. The update of $x_i(k)$ is then as follows:
\begin{eqnarray}
\label{eqn-running-cons-sensor-i}
x_i(k+1) &=& \frac{k}{k+1}  \left( W_{ii}(k)x_i(k)+ \sum_{j \in
O_i(k)}W_{ij}(k) x_j(k) \right)+ \frac{1}{k+1}  \eta_i(k+1), \, k=1,...\\
 x_i(1) &=&  \eta_i(1). \nonumber
\end{eqnarray}
Here $O_i(k)$ is the (random) neighborhood of sensor $i$ at time $k$, and
$ W_{ij}(k)$ are the (random) averaging weights.\footnote{We remark that,
 to implement the algorithm, sensor $i$ has to know the quantities $v_i:=\left[S^{-1}(m_1-m_0)\right]_i$, $[m_1]_i$ and $[m_0]_i$; this
   knowledge can be acquired in the training period of the sensor network.} The local sensor $i$'s decision test at time $k$, $T_{k,i}$, is given by:
\begin{equation}
\label{eqn-decision-x-i}
T_{k,i}:=\mathcal{I}_{\{x_i(k)>0\}},
\end{equation}
i.e., $H_1$ (resp. $H_0$) is decided when $x_i(k)>0$ (resp. $x_i(k) \leq 0$.)
Let
$x(k) =
(x_1(k),x_2(k),...,x_N(k))^\top$ and
$\eta(k)=(\eta_1(k),...,\eta_N(k))^\top$.
Also, collect the averaging weights $W_{ij}(k)$ in $N \times N$
  matrix $W(k)$, where, clearly, $W_{ij}(k)=0$ if the sensors
   $i$ and $j$ do not communicate at time step $k$.
The algorithm in matrix form becomes:
\begin{eqnarray}
\label{eqn_recursive_algorithm}
x(k+1) &=& \frac{k}{k+1} W(k) x(k) + \frac{1}{k+1}  \eta(k+1),\, k=1,...\\
 x(1)  &=&   \eta(1). \nonumber
\end{eqnarray}

We remark that the algorithm in~\eqref{eqn_recursive_algorithm}
extends the running consensus algorithm in~\cite{running-consensus} for spatially correlated
sensor observations (non-diagonal covariance matrix $S$.) When $S$ is diagonal,
the algorithm in~\eqref{eqn_recursive_algorithm} reduces to the algorithm in~\cite{running-consensus}.\footnote{Another
minor difference between~\cite{running-consensus} and eqn.~\eqref{eqn_recursive_algorithm}
is that~\cite{running-consensus} multiplies the log-likelihood ratio term (the term analogous to $\frac{1}{k+1}  \eta(k+1)$)
by $N$; this multiplication does not affect detection performance.}

We allow the averaging matrices
 $W(k)$ to be random. Formally,
 let $(\Omega, \mathcal{F}, \mathbb{P})$
 be a probability space (where $\Omega$
  is a sample space, $\mathcal{F}$
  is a $\sigma$-algebra, and $\mathbb{P}:\mathcal{F} \rightarrow [0,1]$
   is a probability measure.)
    For any $k$,
     $W(k)$
      is a random variable, i.e.,
      an $\mathcal{F}$-measurable
      function $W(k)=W(\omega;k)$, $\omega \in \Omega$, $W(k): \Omega \rightarrow {\mathbb R}^{N \times N}$.
       We now summarize the assumptions on $W(k)$. Recall that $J:=\frac{1}{N}(11^\top)$ and denote by $\widetilde{W}(k):=W(k)-J$. From now on, we will drop the index $k$ from $W(k)$ and $\widetilde{W}(k)$ when
   we refer to the distribution of $W(k)$ and $\widetilde{W}(k)$.
\begin{assumption}
\label{assumption-W(k)}
For the sequence of matrices $\left\{ W(k)  \right\}_{k=1}^{\infty}$, we
assume the following:
\begin{enumerate}
\item The sequence $\left\{ W(k)  \right\}_{k=1}^{\infty}$ is i.i.d.
\item $W$ is symmetric and stochastic (row-sums are equal to 1 and the entries are
nonnegative,) with probability one.
\item The random matrix $W(l)$ and the random vector $y(s)$ are
independent, $\forall l$, $\forall s$.
\end{enumerate}
\end{assumption}
In sections V and VI, we examine
 what (additional) conditions the matrices $W(k)$
  have to satisfy, to achieve asymptotic optimality of
   the distributed detection algorithm.

\mypar{Network supergraph} Define also the network supergraph
 as a pair $G:=(\mathcal{V}, E)$, where $\mathcal{V}$ is the set of nodes with
 cardinality $|\mathcal{V}|=N$, and $E$ is the set of edges with cardinality
 $|E|=M$, defined by: $E=\{ \{i,j\}:\,\, \mathbb{P}\left( W_{ij} \neq 0\right) > 0,\,\,i< j\}.$
  Clearly, when, for some $\{i,j\}$, $\mathbb{P} \left( W_{ij} \neq 0\right) =0 $,
  then the link $\{i,j\} \notin E$ and nodes $i$ and $j$ never communicate.

   %It can be shown that Assumption~\ref{assumption-W(k)}-4) is equivalent to
%    the requirement that the supergraph $G$ is connected. Connectedness
%     of the supergraph means that the underlying network is connected on average;
%      network instantiations during algorithm runs, however, do not have to be connected, like with the gossip algorithm~\cite{BoydGossip}, where, at a time $k$, only one link is active.

   For subsequent analysis, it will be useful
    to define the matrices $\Phi(k,j)$, for $k>j \geq 1$, as follows:
\begin{equation}
\label{eqn-def-Phi}
\Phi(k,j):=W(k-1)W(k-2)...W(j).
\end{equation}
Then, the algorithm in eqn.~\eqref{eqn_recursive_algorithm} can be
written as:
\begin{equation}
\label{eqn_x(k)-equation}
x(k) = \frac{1}{k} \sum_{j=1}^{k-1} \Phi(k,j) \eta(j) + \frac{1}{k}
\eta(k),\,\,k=2,3,...
\end{equation}
Also, introduce:
\begin{equation}
\label{eqn-tilde-phi}
\widetilde{\Phi}(k,j):=\widetilde{W}(k-1) \widetilde{W}(k-2) ... \widetilde{W}(j), \,\,k > j \geq 1,
\end{equation}
and remark that
\[
\widetilde{\Phi}(k,j) = \Phi(k,j)-J.
\]

Recall the definition of the $N\times 1$ vector $v$ in~\eqref{eqn-def-v}.
The  sequence of $N\times 1$ random vectors $\{\eta(k)\}$, conditioned on
$H_l$, is i.i.d. The vector $\eta(k)$ (under hypothesis $H_l$, $l=0,1$) is Gaussian with mean
$m_{\eta}^{(l)}$ and
covariance $S^{\eta}$:
\begin{eqnarray}
\label{eqn_mu_sigma}
m_{\eta}^{(l)} &=& (-1)^{(l+1)} \mathrm{Diag} \left( v
\right)\,\frac{1}{2}(m_1-m_0)\\
S^{\eta} &=& \mathrm{Diag} \left( v\right)  S  \mathrm{Diag} \left( v\right).
\end{eqnarray}
Here $\mathrm{Diag}(v)$ is a diagonal matrix with the diagonal entries
equal to the entries of $v$.

\section{Asymptotic performance of distributed detection: Switching fusion example}
\label{section-example}
In this section, we examine asymptotic performance of distributed detection algorithm on a simple and impractical,
 yet illustrative example; we tackle the generic case in Section~{VI}. The network at a time step $k$ can either be fully connected, with probability $p$,
 or completely disconnected (without edges,) with probability $1-p$. Specifically, the distribution of the random
 averaging matrix $W(k)$ is given by:
  \begin{equation}
  \label{eqn-W-distr}
  W(k) = \left\{ \begin{array}{ll}
 J &\mbox{ with prob. $p$} \\
  I &\mbox{ with prob. $1-p$.}
  \end{array} \right.
  \end{equation}
With model~\eqref{eqn-W-distr}, at each time step $k$, each sensor
behaves as a fusion node, with probability $p$, and as an individual detector, with probability $1-p$.
 We call this communication model the switching fusion. We show that, to achieve asymptotic optimality of distributed detection,
 the fusion step ($W(k)=J$) should occur sufficiently often, i.e., $p$ should exceed a threshold.
 Namely, we find necessary and sufficient condition for the asymptotic optimality in terms of $p$. When distributed detection is not optimal ($p$ is below the threshold,) we find the achievable
  rate of decay of the error probability, as a function of $p$. The goal of the switching fusion example is two-fold.
   First, it provides insight on how the amount of communication (measured by $p$) affects detection performance. Second, it explains in a clear and natural way our methodology for quantifying detection performance
   on generic networks (in Section~{VI}.) Namely, Section~{VI} mimics and extends the analysis from Section~{V} to derive distributed detection performance on generic networks. We next detail the
    sensor observations model.

  We assume that the observations $y_i(k)$ of $N$ different sensors
  are uncorrelated, and that the individual Chernoff information, given by eqn.~\eqref{C-tot-uncorr}, is the same at each sensor $i$.
    Hence, we have ${\bf{C}_{\mathrm{tot}}}=N\,\bf{C_i}$.
     We assume that, at time instant $k$, the network can either be fully connected, with
 probability $p$, or without edges, with probability $1-p$.

 Denote by $P^e_{i,\mathrm{dis}} (k)$ the Bayes error probability at sensor $i$, after $k$
 samples are processed. We have the following Theorem on the asymptotic performance of the distributed detection algorithm.
\begin{theorem}
\label{theorem-neces-suf-special-case}
Consider the distributed detection algorithm given by eqns. \eqref{eqn-running-cons-sensor-i} and \eqref{eqn-decision-x-i}. Assume that
 the sensor observations are spatially uncorrelated and
 that the Chernoff information ${\bf{C_i}}$ is equal at each sensor $i$. Let $W(k)$
 be i.i.d. matrices with the distribution given by eqn.~\eqref{eqn-W-distr}. Then,
 the exponential decay rate of the error probability is given by:

 \begin{eqnarray}
 \label{eqn-teorema-rate}
\lim_{k \rightarrow \infty}\,-\frac{1}{k} \log P^e_{i,\mathrm{dis}}(k) =  \left\{ \begin{array}{ll}
{\bf{C_{\mathrm{tot}}}} &\mbox{ if $|\log(1-p)| \geq {\bf{C_{\mathrm{tot}}}}(N-1)$} \\
  {{\bf{C_i}}} +  |\log(1-p)|  &\mbox{ if $|\log(1-p)| \leq \frac{{\bf{C_{\mathrm{tot}}}}(N-1)}{N^2}$} \\
  2 \sqrt{\frac{|\log(1-p)|{\bf{C_{\mathrm{tot}}}}}{N-1}} - \frac{|\log(1-p)|}{N-1}  &\mbox{ otherwise.}
       \end{array} \right.
\end{eqnarray}

 Moreover, a necessary and sufficient condition for
  asymptotic optimality, in the sense of Definition~\ref{def-asym-optimal}, is given by:
  \begin{equation}
  \label{eqn-necess-suff-special-case}
  \frac{|\log(1-p)|}{N-1} \geq {\bf{C_{\mathrm{tot}}}}=N {\bf{C_i}}.
  \end{equation}
\end{theorem}

Condition~\eqref{eqn-necess-suff-special-case} says that the network connectivity should be
  good enough (i.e., $p$ should be large enough,) in order to achieve the asymptotic optimality
   of distributed detection. Also, there is a ``phase change'' behavior, in a sense that
    distributed detection is asymptotically optimal above a threshold on $p$, and
    it is not optimal below that threshold. Further, we can see that, as
    $p$ decreases, distributed detection performance becomes worse and worse, and it
    approaches the performance of an individual sensor-detector. (See Figure 1, and eqn. \eqref{eqn-teorema-rate}.)

   %Further, the quantity $\frac{1}{N-1}|\log(1-p)|$ may be interpreted
%   as the information capacity of the network. If the network has
%   the information capacity greater than the total Chernoff information offered by
%   sensors' observations, then the asymptotic optimality is achieved; conversely,
%   distributed detection cannot be optimal if the
%   network's information capacity is lower than the total Chernoff information.
%    Further, $p=1$ corresponds to the fully connected network, and hence, to the optimal centralized detector.
%     Capacity of the fully connected network (that has $W(k) \equiv J$) is $+\infty$,
%      and asymptotic optimality is achieved for any amount of Chernoff information.

      We proceed with proving Theorem~\ref{theorem-neces-suf-special-case}. In Section {VI}, we will
      follow a reasoning similar to the proof of Theorem~\ref{theorem-neces-suf-special-case}
       to provide a sufficient condition for asymptotic optimality on generic networks.

       \begin{figure}[thpb]
      \centering
      \includegraphics[height=2.
      in,width=3.1in]{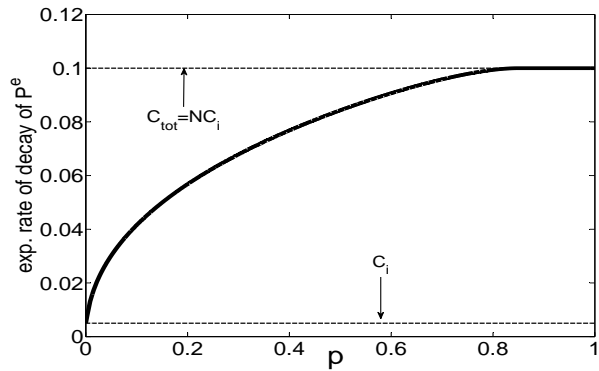}
      %\includegraphics[height=2.5
      %in,width=3.62in]{slika_rate_p_3.pdf}
      \caption{ Exponential decay rate of error probability $\phi^\star$ (given by eqn. (26))
       for the model considered in Section \ref{section-example}. The network has $N=20$ sensors and
       ${\bf{C_{\mathrm{tot}}}}=0.1$. The optimal rate (equal to ${\bf{C_{\mathrm{tot}}}}$) is achieved for
       $p \geq 0.83$.}
      \label{Figure_rate_vs_p}
\end{figure}

\begin{proof}[Proof of Theorem~\ref{theorem-neces-suf-special-case}]
First, remark that $x(k)$, conditioned on $H_0$,
is equal in distribution to $-x(k)$, conditioned on $H_1$. This is true
 because $\eta(k)$, conditioned on $H_0$, is equal in distribution
  to $-\eta(k)$, conditioned on $H_1$, for all $k$; and the distribution of $W(k)$
   does not depend on the active hypothesis, $H_0$ or $H_1$.
   Denote by $\mathbb{P}_l\left( \cdot\right)=\mathbb{P}(\cdot|H_l)$, $l=0,1$, and
   consider the probability of false alarm,
   the probability of miss, and the Bayes error probability at sensor $i$
    (with the running consensus detector,) respectively, given by:
    \begin{eqnarray}
    \alpha_{i,\mathrm{dis}}(k) &=& \mathbb{P}_0 \left(  x_i(k)>0\right),\,\,
    \beta_{i,\mathrm{dis}}(k) = \mathbb{P}_1 \left(  x_i(k)\leq 0\right)\\
    \label{eqn-p-e-dis}
    P^e_{i,\mathrm{dis}}(k) &=& P(H_0)\mathbb{P}_0 \left(  x_i(k)>0\right) + P(H_1) \mathbb{P}_1 \left(  x_i(k)\leq 0\right).
    \end{eqnarray}
   Remark that $\mathbb{P}_l \left( x_i(k)=0\right)=0$, $l=0,1$.
   Thus, we have that
   \begin{eqnarray}
   \label{eqn-equal-alpha-beta}
   \beta_{i,\mathrm{dis}}(k) &=& \mathbb{P}_1  \left( x_i(k) \leq 0\right) = \mathbb{P}_0 \left(  x_i(k)>0\right) = \alpha_{i,\mathrm{dis}}(k), \,\, \forall k,\,\,\forall i.
   \end{eqnarray}
   From eqns. \eqref{eqn-p-e-dis} and \eqref{eqn-equal-alpha-beta}, it can be shown that:
   \begin{eqnarray}
   \label{eqn-limsup-p-e}
   \limsup_{k \rightarrow \infty} \frac{1}{k} \log P^e_{i,\mathrm{dis}}(k)  &=&  \limsup_{k \rightarrow \infty} \frac{1}{k} \log \alpha_{i,\mathrm{dis}}(k)\\
   \label{eqn-liminf-p-e}
   \liminf_{k \rightarrow \infty} \frac{1}{k} \log P^e_{i,\mathrm{dis}}(k)  &=&  \liminf_{k \rightarrow \infty} \frac{1}{k} \log \alpha_{i,\mathrm{dis}}(k).
   \end{eqnarray}

   We further assume that $H_0$ is true, and we restrict our attention to $\alpha_{i,\mathrm{dis}}(k)$,
    but the same conclusions (from \eqref{eqn-equal-alpha-beta}) will be valid for $\beta_{i,\mathrm{dis}}(k)$ also.
     We now make the key step in proving Theorem~\ref{theorem-neces-suf-special-case},
     by defining a partition of the probability space $\Omega$. Fix the time step $k$ and denote by
 $A_l$, $l=0,...,k-1$, the event
 \[
 A_l =
 \left\{ \begin{array}{ll}
 \left\{ \max \left\{{s \in\{1,...,k-1\}}:\,\,W(s)=J \right\}=l  \right\} &\mbox{ for $l=1,...,k-1$} \\
 \{W(s)=I,\,s=1,...,k-1\}&\mbox{ for $l=0.$}
       \end{array} \right.
\]
That is, $A_l$ is the event that the largest time step $s \leq k-1$, for which
 $W(s)=J$, is equal to $s=l$. (The event $A_l$ includes the scenarios of arbitrary realizations of $W(s)$--either $J$ or $I$--for
 $s\leq l$; but it requires $W(s)=I$ for all $l<s \leq k-1$.) Remark that $A_l$ is a function of $k$, but the dependence on $k$ is dropped for
notation simplicity. We have that $\mathbb{P}  \left(  A_l \right) =  p(1-p)^{k-l-1}$, for $l=1,...,k-1$,
 and $\mathbb{P}\left( A_0\right)=(1-p)^{k-1}$. Also,
  each two events, $A_l$ and $A_j$, $j \neq l$, are disjoint, and
   $\cup_{l=0}^{k-1}\,A_l=\Omega,$ i.e., the events
    $A_l$, $l=0,...,k-1$, constitute a finite partition of
    the probability space $\Omega$. (Note that $\sum_{l=0}^{k-1}\mathbb{P}(A_l)=1$.) Recall the definition of
    $\Phi(k,j)$ in eqn.~\eqref{eqn-def-Phi} and note that, if $A_l$ occurred,
we have:
 \begin{equation}
  \Phi(k,s) = \left\{ \begin{array}{ll}
 I &\mbox{ if $k-1 \geq s>l$} \\
  J &\mbox{ if $s \leq l$.}
       \end{array} \right.
  \end{equation}
Further, conditioned on $A_l$,
we have that (when $l=0$, first sum in eqn. \eqref{eqn-prva-druga} does not exist)
\begin{eqnarray}
\label{eqn-prva-druga}
x(k) &= &   \frac{1}{k} \left( \sum_{j=1}^l J \eta(j) + \sum_{j=l+1}^{k} I \eta(j) \right)\\
&=&         \frac{1}{k} \left( \sum_{j=1}^l \left( \frac{1^\top \eta(j)}{N} \right)1 + \sum_{j=l+1}^{k}\eta(j)  \right).\nonumber
\end{eqnarray}
Hence, conditioned on $A_l$, $x_i(k)$
 is a Gaussian random variable,
 \[x_i(k)|A_l \sim \mathcal{N}  \left( \theta(l;k),\,\zeta^2(l;k) \right),\]
 where
 \begin{eqnarray}
 \theta(l;k) &=&  -\frac{4l}{N}{\bf{C_{\mathrm{tot}}}} -4(k-l) {\bf{C_i}} = -4 k {\bf{C_i}}\\
 \zeta^2(l;k)  &=&  \frac{8l}{N^2} {\bf{C_{\mathrm{tot}}}} + 8(k-l) {\bf{C_i}}.
  \end{eqnarray}
Define
\begin{equation}
\label{eqn-def-chi}
\chi(l;k):=-\frac{\theta(l;k)}{\zeta(l;k)} = \frac { \sqrt{2{\bf{C_i}}} {k} } {  \sqrt{ \frac{l}{N} + (k-l)} },
\end{equation}
and remark that
\[
\mathbb{P}_0 \left( x_i(k)>0  |A_l \right) = \mathcal{Q}(\chi(l;k)),
\]
where $\mathbb{P}_0(\cdot) := \mathbb{P}(\cdot|H_0)$.
%where
%\begin{eqnarray}
%f_i(l;k) &=& \frac{ \left( l {\bf{C_{\mathrm{tot}}}} + N(k-l) {\bf{C_i}}  \right)^2}
%{   l {\bf{C_{\mathrm{tot}}}} + N^2 (k-l) C_i     }\\
%&=&  \frac{k N  {\bf{C_i}} }   { 1 + \frac{(N-1)(k-l)}{k}    }.
%\end{eqnarray}

Using the total probability, we can write $\alpha_{i,\mathrm{dis}}(k)=\mathbb{P}_0 \left( x_i(k)>0  \right)$ as:
\begin{eqnarray}
\label{eqn-alpha-total-prob-law}
\alpha_{i,\mathrm{dis}}(k) &=&  \sum_{l=0}^{k-1} \mathbb{P}_0\left( x_i(k)>0 | A_l \right)\, \mathbb{P}(A_l)\\
&=&\sum_{l=1}^{k-1}  \mathcal{Q}(\chi(l;k)) \,p(1-p)^{k-l-1}+ \mathcal{Q}(\chi(0;k)) \,(1-p)^{k-1}. \nonumber
\end{eqnarray}
%where $\phi(j;k):= \frac{ {\bf{C_{\mathrm{tot}}}}  }{1 + (N-1)\frac{j}{k}}+ \frac{j}{k} |\log(1-p)|$.

We now proceed with calculating the exponential rate of decay of
$\alpha_{i,\mathrm{dis}}(k)$ (and hence, $P^e_{i,\mathrm{dis}}(k)$) as $k \rightarrow \infty$. The key ingredient to do that is the representation of $\alpha_{i,\mathrm{dis}}(k)$ in eqn. \eqref{eqn-alpha-total-prob-law}, and the inequalities for the $\mathcal{Q}$-function in eqn. \eqref{eqn-Q-ineq}.
 Namely, it can be shown that, when $k$ grows large, $\alpha_{i,\mathrm{dis}}(k)$, and hence, $P^e_{i,\mathrm{dis}}(k)$, behaves as (here we present the main
 idea but precise statements are in the Apendix):
 \begin{eqnarray}
 P^e_{i,\mathrm{dis}}(k) \sim \sum_{j=0}^{k-1} e^{-k\,\phi(j;k)},
 \end{eqnarray}
where
\begin{eqnarray}
\label{eqn-def-phi-j-k}
\phi(j;k):= \frac{{\bf{C_{\mathrm{tot}}}}} {1+(N-1)\frac{j+1}{k}} + \frac{j}{k} |\log(1-p)| .
\end{eqnarray}
Hence, the quantities $\phi(j;k)$, for different $j$'s, represent different ``modes''
 of decay; the decay of $ P^e_{i,\mathrm{dis}}(k)$ is then determined by the slowest mode $\widehat{\phi}(k)$, defined by:
\begin{equation}
\widehat{\phi}(k) = \min_{j=0,...,k-1} \phi(j;k) .
\end{equation}
More precisely, using eqns. \eqref{eqn-limsup-p-e}, \eqref{eqn-liminf-p-e},
the expression for $\alpha_{i,\mathrm{dis}}(k)$ in eqn. \eqref{eqn-alpha-total-prob-law}, and
 the inequalities \eqref{eqn-Q-ineq}, it can be shown that:
\begin{eqnarray}
\label{eqn-hat-phi-limsup}
\liminf_{k \rightarrow \infty} \,-\frac{1}{k} \log P^e_{i,\mathrm{dis}}(k) &\geq&  \liminf_{k \rightarrow \infty} \widehat{\phi}(k)\\
\limsup_{k \rightarrow \infty} \,-\frac{1}{k} \log P^e_{i,\mathrm{dis}}(k) &\leq&  \limsup_{k \rightarrow \infty} \widehat{\phi}(k). \nonumber \nonumber
\end{eqnarray}
The detailed proof of inequalities \eqref{eqn-hat-phi-limsup} is in the Appendix.

We proceed by noting that the minimum of $\phi(j;k)$ over
the discrete set $j \in \{0,1,...,k-1\}$ does not differ much
from the minimum of $\phi(j;k)$ over the interval $[0,k-1]$.
Denote by $\phi^\star(k)$ the minimum of $\phi(j;k)$ over $[0,k-1]$:
\begin{eqnarray}
\label{eqn-hat-phi}
\phi^\star(k)  &=&  \min_{j \in [0,k-1]} \phi(j;k).
\end{eqnarray}

Then, it is easy to verify that:
\begin{eqnarray}
\label{eqn-ineq-phi}
\phi^\star(k) \leq \widehat{\phi}(k) \leq \phi^\star(k)\,(1+\frac{N-1}{k})+\frac{|\log(1-p)|}{k}.
\end{eqnarray}

The function $\phi(j,k)$ is convex in its first argument on $j \in [0,k-1]$; it is straightforward to calculate $\phi^\star(k)$, which can be shown to be equal to:

\begin{eqnarray}
\label{eqn-phi-star}
\phi^\star(k) =  \left\{ \begin{array}{ll}
\frac{{\bf{C_{\mathrm{tot}}}}}{1+\frac{N-1}{k}} &\mbox{ if $|\log(1-p)| \geq \frac{{\bf{C_{\mathrm{tot}}}} (N-1)}{(1+\frac{N-1}{k})^2}$} \\
  \frac{{\bf{C_{\mathrm{tot}}}}}{N} + \frac{k-1}{k} |\log(1-p)|  &\mbox{ if $|\log(1-p)| \leq \frac{{\bf{C_{\mathrm{tot}}}}(N-1)}{N^2}$} \\
  2 \sqrt{\frac{|\log(1-p)|{\bf{C_{\mathrm{tot}}}}}{N-1}} - \frac{|\log(1-p)|}{N-1} - \frac{|\log(1-p)|}{k} &\mbox{ otherwise.}
       \end{array} \right.
\end{eqnarray}

The limit $\lim_{k \rightarrow \infty} \phi^\star(k)=: \phi^\star$ exists, and is equal to:

\begin{eqnarray}
\label{eqn-phi-s}
\phi^\star =  \left\{ \begin{array}{ll}
{\bf{C_{\mathrm{tot}}}} &\mbox{ if $|\log(1-p)| \geq {\bf{C_{\mathrm{tot}}}}(N-1)$} \\
  \frac{{\bf{C_{\mathrm{tot}}}}}{N} +  |\log(1-p)|  &\mbox{ if $|\log(1-p)| \leq \frac{{\bf{C_{\mathrm{tot}}}}(N-1)}{N^2}$} \\
  2 \sqrt{\frac{|\log(1-p)|{\bf{C_{\mathrm{tot}}}}}{N-1}} - \frac{|\log(1-p)|}{N-1}  &\mbox{ otherwise.}
       \end{array} \right.
\end{eqnarray}

From eqns. \eqref{eqn-ineq-phi} and \eqref{eqn-phi-s}, we have:
\begin{equation}
\label{eqn-sad}
\lim_{k \rightarrow \infty} \widehat{\phi}(k) = \lim_{k \rightarrow \infty} \phi^\star(k) =  \phi^\star.
\end{equation}

In view of eqns. \eqref{eqn-sad} and \eqref{eqn-hat-phi-limsup}, it follows that the
rate of decay of the error probability at sensor~$i$~is:
\begin{eqnarray}
\lim_{k \rightarrow \infty}\,- \frac{1}{k} \log P^e_{i,\mathrm{dis}}(k) =  \lim_{k \rightarrow \infty} \widehat{\phi}(k)
=  \phi^\star.
\end{eqnarray}
The necessary and sufficient condition for asymptotic optimality then follows from eqn. \eqref{eqn-phi-s}.

%Hence, a sufficient condition for asymptotic optimality
%is that
%\[
%\liminf_{k \rightarrow \infty} \min_{j \in \{0,...,k\}} \phi(j;k) \geq {\bf{C_{\mathrm{tot}}}}.
%\]
%Now, by the same reasoning as for the necessary condition, we can see that~\eqref{eqn-necess-suff-special-case} is also a sufficient condition
%for asymptotic optimality.
%%We summarize our findings in the following Theorem.
%%
%
%Hence, a necessary condition for asymptotic optimality
% of the distributed detection algorithm is that
% \[
% \limsup_{k \rightarrow \infty} \min_{j \in \{0,...,k\}} \phi(j;k) \geq  {\bf{C_{\mathrm{tot}}}},
% \]
%
%
%Consider the function $j \mapsto \phi(j;k)$, on $j \in \mathbb R$. The function
%  $\phi(\cdot;k)$ is strictly convex on $\mathbb R$, for all $k$. Further, note that $\phi(j;0) = {\bf{C_{\mathrm{tot}}}}$.
%  Denote $j^\star(k) = \mathrm{argmin}_{j \in \mathbb R} \phi(j;k) $.
%  Thus, a necessary condition for asymptotic optimality
%  is that
%  \begin{equation}
%  \label{eqn-manje-od-1}
%  \liminf_{k \rightarrow \infty} j^\star(k) < 1.
%  \end{equation}
%  It is easy to show that
%  \[
%  j^\star(k) = \frac{k}{N-1} \left( \sqrt{     \frac{ (N-1) {\bf{C_{\mathrm{tot}}}}  }{ |\log (1-p)| }     } -1 \right).
%  \]
%  Hence, a necessary condition for~\eqref{eqn-manje-od-1} to hold is:
%  \begin{equation}
%  \label{eqn-necessary-suff}
%  \frac{(N-1){\bf{C_{\mathrm{tot}}}}}{|\log(1-p)|} \leq 1,
%  \end{equation}
% and the claim on necessity of condition~\eqref{eqn-necess-suff-special-case} follows.
%
%
\end{proof}

\section{Asymptotic performance of distributed detection: General case}
\label{section-asympt-analysi-running-cons}
This section provides a necessary condition, and a sufficient condition for asymptotic optimality of distributed detection on generic networks and for generic, spatially correlated, Gaussian observations.
When distributed detection is not guaranteed to be optimal, this section finds a lower bound on the exponential decay rate of error probability, in terms of the system parameters. We start by pursuing sufficient conditions for optimality and evaluating the lower bound on the decay rate of the error probability.

\subsection{Sufficient condition for asymptotic optimality}
Recall that $r:=\lambda_2\left( \mathbb{E} \left[ W(k)^2 \right] \right) = \|  \mathbb{E} \left[ W(k)^2 \right] -J   \|$.
 It is well known that the quantity $r$ measures the speed of the information flow, i.e., the speed of the averaging
 across the network, like with standard consensus and gossip algorithms, e.g., \cite{BoydGossip}.
 (The smaller $r$ is, the faster the averaging is.) The next Theorem shows that distributed detection is asymptotically optimal if the network information flow is fast enough, i.e., if $r$ is small enough. The Theorem
  also finds a lower bound on the rate of decay of the error probability, even when the sufficient
  condition for asymptotic optimality does not hold. Recall also ${\bf{C_{\mathrm{tot}}}}$
  in eqn.~\eqref{C-tot-uncorr}.

\begin{theorem}
\label{theorem-sufficient-condition}
Let Assumption~\ref{assumption-W(k)} hold and consider the distributed detection algorithm
 defined by eqns. \eqref{eqn-running-cons-sensor-i} and \eqref{eqn-decision-x-i}. Then, the following holds for the exponential decay rate of the error
 probability at each sensor:
 \begin{eqnarray}
 \label{eqn-rate-bound}
 \liminf_{k \rightarrow \infty}\,-\frac{1}{k} \log P^e_{i,\mathrm{dis}}(k)
 \geq
 \left\{ \begin{array}{ll}
  {\bf{C_{\mathrm{tot}}}}  &\mbox{ if $|\log r| \geq \frac{1}{8} N^2  \left( 1+(1-\frac{1}{N})K\right)\|S^{\eta}\| $} \\
  -\left(\frac{1}{2N^2} \sigma_L^2 \overline{\mu}^2 + \frac{1}{N} m_{L}^{(0)}\overline{\mu}\right) &\mbox{ otherwise,}
       \end{array} \right. \nonumber
 \end{eqnarray}
 where
 \begin{equation}
 \label{eqn-mu-overline}
 \overline{\mu} =
 \left\{ \begin{array}{llll}
  \vspace{1mm}
 \frac{1}{4} \frac{K}{K+1} + \frac{1}{4} \frac{  \sqrt{ K^2+\frac{32|\log r|}{\|S^{\eta}\|}\left( 1+K\right) } }  {K+1}, \\
  \vspace{1mm}
 \mathrm{if\,\,}\frac{1}{8}\|S^{\eta}\| < |\log r| < \frac{1}{8}N^2(1+(1-\frac{1}{N}K))\|S^{\eta}\| ;
 \\
  \vspace{1mm}
  \frac{1}{4}\sqrt{K^2+\frac{32|\log r|}{\|S^{\eta}\|}}-\frac{1}{4}K,\\
   \vspace{1mm}
\mathrm{if\,\,} |\log r| \leq  \frac{1}{8}\|S^{\eta}\|.
       \end{array} \right.
 %\frac{1}{4} \frac{K}{K+1} + \frac{1}{4} \frac{  \sqrt{ K^2+\frac{32|\log r|}{\|S^{\eta}\|}\left( 1+K\right) } }  {K+1},
 \end{equation}
 Here $K=\left( 8 \overline{m}\right)/\|S^{\eta}\|$. Moreover, each sensor $i$ is asymptotically optimal, and $ \lim_{k \rightarrow \infty} - \frac{1}{k} \log P^e_{i,\mathrm{dis}}(k) = {\bf{C_{\mathrm{tot}}}},\,\forall i,
 $
 provided that:
 \begin{equation}
 \label{eqn-suf-cond-generic}
 {|\log r|} \geq \frac{1}{8} N^2  \left( 1+(1-\frac{1}{N})K\right)\|S^{\eta}\|.
 \end{equation}

 %
% \begin{eqnarray}
% -\limsup_{k \rightarrow \infty}\frac{1}{k} \log P^e_{i,\mathrm{dis}}(k) \geq  \left\{ \begin{array}{rl}
%  {\bf{C_{\mathrm{tot}}}}  &\mbox{ if $|\log r| \geq \frac{1}{2}N(N-1)\overline{m} +\frac{1}{8}N^2 \|S^{\eta}\|$} \\
%  -\left(\frac{1}{2N^2} \sigma_L^2 \overline{\mu}^2 + \frac{1}{N} m_{L}^{(0)} \overline{\mu}\right) &\mbox{ otherwise,}
%       \end{array} \right.
% \end{eqnarray}
% where
% \begin{equation}
% \label{eqn-mu-overline}
% \overline{\mu} =\frac{\overline{m}+\sqrt{\overline{m}^2 + 4 |\log r| (2 \overline{m}+\frac{1}{2}\|S^{\eta}\|)}}{2(2 \overline{m}+\frac{1}{2}\|S^{\eta}\|)}.
% \end{equation}
% Moreover, sufficient condition for asymptotic optimality (in the sense of Definition~\ref{def-asym-optimal},)
% at each sensor $i$, is given by:
% \begin{equation}
% \label{eqn-suf-cond-generic}
% \frac{|\log r|}{N} \geq \frac{1}{2} N(N-1) \overline{m} + \frac{1}{8} N \|S^{\eta}\|.
% \end{equation}
% Specifically, if the sensor observations are spatially uncorrelated, then~\eqref{eqn-suf-cond-generic}
% becomes:
% \begin{equation}
% \frac{|\log r|}{N} \geq (3N-1) \, \max_{i=1,...,N}\,{\bf{C_i}}.
% \end{equation}
 \end{theorem}

 \begin{figure}[thpb]
      \centering
      \includegraphics[height=2.2
      in,width=3.32in]{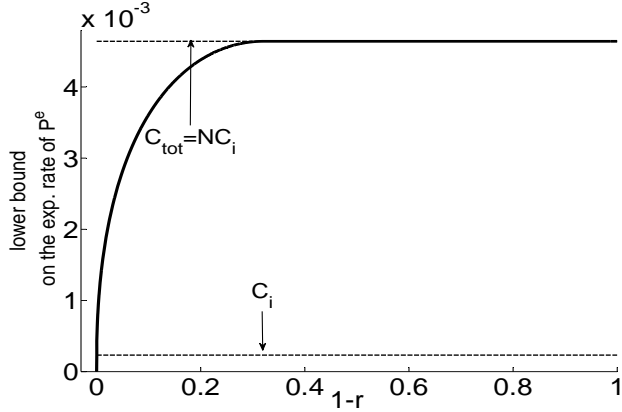}
      \caption{ Lower bound on the exponential decay rate of the error probability (given by eqn. (52))
       for the (generic) model in Section VI. The network has $N=20$ sensors, and
       ${\bf{C_{\mathrm{tot}}}}=0.0047$. The optimal rate (equal to ${\bf{C_{\mathrm{tot}}}}$) is achieved for
       $p \geq 0.29$.}
      \label{Figure_rate_vs_p_generic}
\end{figure}

We proceed by proving Theorem~\ref{theorem-sufficient-condition}. We first set up the proof and state some auxiliary
 Lemmas. We first consider the probability of false alarm, $\alpha_{i,\mathrm{dis}}(k)=\mathbb{P}_0\left( x_i(k)>0  \right)$,
but the same conclusions will hold for $\beta_{i,\mathrm{dis}}(k)=\mathbb{P}_1\left( x_i(k)<0  \right)$ also
 (See the Proof of Theorem~\ref{theorem-neces-suf-special-case}.) We examine the family of Chernoff bounds on the probability of false alarm, parametrized by $\mu > 0$, given by:
 \begin{equation}
 \label{eqn-C-bound}
 \alpha_{i,\mathrm{dis}}(k)\leq \mathbb{E}_0 \left[ e^{k \mu x_i(k)} \right]
 =: \mathcal{C}(k\mu),
 \end{equation}
 where $\mathbb{E}_l\left[ a\right]:=\mathbb{E}\left[ a|H_l\right]$, $l=0,1.$ We then examine the conditions under which the best Chernoff bound falls below the negative of the Chernoff information, in the limit as $k \rightarrow \infty$. More precisely, we
 examine under which conditions the following inequality holds:
 \begin{equation}
 \limsup_{k \rightarrow \infty} \frac{1}{k} \log \mathcal{C}(k \mu) \leq -{\bf{C}_{\mathrm{tot}}}.
 \end{equation}
 In subsequent analysis, we will use the following Lemmas, Lemma~\ref{lemma-stoch-matr} and Lemma~\ref{lemma-phi-epsilon}; proof of
 Lemma~\ref{lemma-stoch-matr} is trivial, while proof of Lemma~\ref{lemma-phi-epsilon} is in the Appendix.

\begin{lemma}
\label{lemma-stoch-matr}
Let $V$ be a $N \times N$ stochastic matrix and consider the matrix $\widetilde{V}=V-J$. Then, for all $i=1,...,N$, the following inequalities hold:
\begin{eqnarray}
\sum_{l=1}^N |\widetilde{V}_{il}| &\leq& 2\frac{N-1}{N}<2\\
\sum_{l=1}^N |\widetilde{V}_{il}|^2 &\leq&  \frac{N-1}{N} <1.
\end{eqnarray}
\end{lemma}
%\begin{proof}
%The proof is trivial.
%\end{proof}

%Define $r:=\rho \left(  W(k)^2-J \right)$.

\begin{lemma}
\label{lemma-phi-epsilon}
Let Assumption~\ref{assumption-W(k)} hold. Then, the following inequality holds:
\begin{equation}
\mathbb{P} \left( \| \widetilde{\Phi}(k,j)  \|  > \epsilon \right) \leq \frac{N^4}{\epsilon^2} r^{k-j}.
\end{equation}
\end{lemma}
%\begin{proof}
%See Appendix.
%\end{proof}

When proving Theorem~\ref{theorem-sufficient-condition}, we will
partition the time interval from the first to $(k-1)$-th time step in windows of width $B$,
 for some integer $B\geq 1$. That is, we consider
 the subsets of consecutive time steps $\{k-B,k-B+1,...,k-1\}$, $\{k-2B,...,k-B-1\}$,
 $...$, $\{k-(\mathcal{J}-1)B,...,k-\mathcal{J}B-1\}$, $\{ 1,2,...,k-\mathcal{J} B-1\}$, where $\mathcal{J}$ is the integer part of $\frac{k-1}{B}$. (Note that the total number of these subsets is $\mathcal{J}+1$; each of these subsets contains $B$ time steps,
 except $\{ 1,2,...,k-\mathcal{J} B-1\}$ that in general has the number of time steps less or equal to $B$.) We then define the events $\mathcal{A}_j$,
  $j=1,...,\mathcal{J}+1$, as follows:
\begin{eqnarray}
\mathcal{A}_1 &=& \left\{ \|\widetilde{\Phi}(k,k-B)\| \leq \epsilon \right\}\\
\mathcal{A}_2 &=& \left\{ \|\widetilde{\Phi}(k,k-B)\| > \epsilon \right\} \cap
\left\{  \|\widetilde{\Phi}(k,k-2B)\| \leq \epsilon  \right\}\nonumber\\
\vdots\nonumber\\
\mathcal{A}_{j+1} &=& \left\{ \|\widetilde{\Phi}(k,k-j B)\| > \epsilon \right\} \cap
\left\{  \|\widetilde{\Phi}(k,k-(j+1)B)\| \leq \epsilon  \right\} \nonumber\\
\vdots\nonumber\\
\mathcal{A}_{\mathcal{J}} &=& \left\{ \|\widetilde{\Phi}(k,k-(\mathcal{J}-1) B)\| > \epsilon \right\} \cap
\left\{  \|\widetilde{\Phi}(k,k-\mathcal{J} B)\| \leq \epsilon  \right\} \nonumber \\
\mathcal{A}_{\mathcal{J}+1} &=& \left\{ \|\widetilde{\Phi}(k,k-\mathcal{J} B)\| > \epsilon \right\} .\nonumber
\end{eqnarray}
It is easy to see that the events $\mathcal{A}_j$, $j=1,...,\mathcal{J}+1$, constitute a
 finite partition of the probability space $\Omega$ (i.e.,
  each $\mathcal{A}_i$ and $\mathcal{A}_j$, $i \neq j$,
   are disjoint, and the union of $\mathcal{A}_j$'s, $j=1,...,\mathcal{J}+1$,
   is $\Omega$.)

By Lemma \ref{lemma-phi-epsilon}, the probability of the event $\mathcal{A}_j$, $j=1,...,\mathcal{J}+1$, is bounded from above as follows:
\begin{equation}
\label{eqn-P_A_j-upper-bound}
\mathbb{P} \left( \mathcal{A}_j\right) \leq \frac{N^4}{\epsilon^2 } r^{(j-1)B}.
\end{equation}

We will next explain the idea behind using the partition $\mathcal{A}_j$, $j=1,...,\mathcal{J}+1$.
First, we state Lemma \ref{lemma-event-a-l} that helps to understand this idea.
\begin{lemma}
\label{lemma-event-a-l}
Consider the random matrices $\widetilde{\Phi}(k,j)$ given by eqn. \eqref{eqn-tilde-phi}, and fix two
time steps, $k$ and $j$, $k>j$. Suppose that the event $\mathcal{A}_j$ occurred.
 Then, for any $s \leq k-jB$, and for $i=1,...,N$:
\begin{eqnarray}
\label{eqn-prva}
\sum_{l=1}^N |[\widetilde{\Phi}(k,s)]_{il}| &\leq& N \sqrt{N} \epsilon\\
\label{eqn-druga}
\sum_{l=1}^N |[\widetilde{\Phi}(k,s)]_{il}|^2 &\leq& N \epsilon^2.
\end{eqnarray}
\end{lemma}
\begin{proof}
See the Appendix.
\end{proof}

Consider now the partition $\mathcal{A}_j$, $j=1,...,\mathcal{J}+1$. Given that $\mathcal{A}_j$
 occurred, the matrices $\Phi(k,s)$, for $s \leq k-jB$, are
 ``$\epsilon$-close'' to $J$ (in view of Lemma \ref{lemma-event-a-l}.) For $s>k-jB$,
  $\Phi(k,s)$ is not ``$\epsilon$-close'' to $J$, and thus
   we pass to the ``worst case'' as given by Lemma \ref{lemma-stoch-matr}. (We may think of
   this as setting $\Phi(k,s)\approx I$.) Hence, given that $\mathcal{A}_j$
    occurred, $x(k)$, in a crude approximation, behaves as:
    \[
    x(k) \sim \frac{1}{k} \sum_{l=1}^{k-jB} J \eta(l) + \frac{1}{k} \sum_{l=k-jB+1}^k I \eta(l).
    \]
Then, by conditioning $x(k)$ on $\mathcal{A}_j$, $j=1,...,\mathcal{J}+1$,
and by using the total probability (as with Theorem \ref{theorem-neces-suf-special-case},)
 we express $P^e_{i,\mathrm{dis}}(k)$ in terms of the $j=1,...,\mathcal{J}+1$ ``modes of decay''
  that correspond to $\mathcal{A}_j$, $j=1,...,\mathcal{J}+1$. Finally, we find the slowest
  among the $\mathcal{J}+1$ modes, as with Theorem \ref{theorem-neces-suf-special-case}.
   The difference compared to Theorem \ref{theorem-neces-suf-special-case}, however,
    is that here we work with the Chernoff bounds on $\alpha_{i,\mathrm{dis}}(k)$ and $P^e_{i,\mathrm{dis}}(k)$, rather than
     directly with $\alpha_{i,\mathrm{dis}}(k)$ and $P^e_{i,\mathrm{dis}}(k)$.
     After this qualitative description and after explaining the main idea behind the calculations, we proceed with
     the exact calculations, some of which are moved to the Appendix.

\begin{proof}[Proof of Theorem~\ref{theorem-sufficient-condition}] Consider the Chernoff bound on $\alpha_{i,\mathrm{dis}}(k)$ given by eqn.~\eqref{eqn-C-bound}, and denote by $\overline{m}_0:=\max_{i=1,...,N}|[m_{\eta}^{(0)}]_i|$.
 We first express the Chernoff upper bound on $\limsup_{k \rightarrow \infty} \frac{1}{k} \log P^e_{i,\mathrm{dis}}(k) $
 in terms of the ``modes of decay,'' similarly as in the Proof of Theorem~\ref{theorem-neces-suf-special-case}. Namely,
  it can be shown (details are in the Appendix) that the following inequality holds:
\begin{eqnarray}
\label{eqn-append}
\limsup_{k \rightarrow \infty} \frac{1}{k} \log \alpha_{i,\mathrm{dis}}(k) &=&
\limsup_{k \rightarrow \infty} \frac{1}{k} \log P^e_{i,\mathrm{dis}}(k)  \\
&\leq &
\limsup_{k \rightarrow \infty} \frac{1}{k} \log \overline{\mathcal{C}}(k\mu) \overline{\delta}(k \mu), \,\,\forall \mu>0,\,\forall \epsilon \in (0,1),\,\,\forall B=1,2,... \nonumber
\end{eqnarray}
where
\begin{eqnarray}
\label{eqn-C-cal-over}
\overline{\mathcal{C}}(k\mu) &=& \mathrm{exp} \left((k-1)(\frac{1}{2N^2} \sigma_L^2 \mu^2 + \frac{1}{N} m_L^{(0)}\,\mu) \right) \\
\label{eqn-delta-over}
\overline{\delta}(k\mu) &=&  \sum_{j=0}^{\mathcal{J}} \mathrm{exp} \left( |2 \mu^2-\mu|\,
\overline{m} \, \left\{2(j+1)B + (k-(j+1)B) N\sqrt{N}\epsilon \right\} \right)\\
&\,&\mathrm{exp} \left( \frac{\mu^2}{2} \,\|S^{\eta}\| \,
\left\{  (j+1)B + (k-(j+1)B)N\epsilon^2\right\}\right)  \left(  \frac{N^4\,r^{jB}}{\epsilon^2}\right). \nonumber
\end{eqnarray}
Notice that the summands in eqn. \eqref{eqn-delta-over} are the ``modes of decay'' that correspond to the events $\mathcal{A}_{j+1}$, $j=1,...,\mathcal{J}+1$. Next, after bounding from above by
``the slowest mode,'' we get:
\begin{eqnarray}
\label{eqn-max}
\overline{\delta}(k \mu) &\leq&  \left( \mathcal{J}+1 \right) \,\max_{j=0,...,\mathcal{J}}\, \mathrm{exp} \left( |2 \mu^2-\mu|\,
\overline{m} \, \left\{2(j+1)B + (k-(j+1)B) N\sqrt{N} \epsilon \right\} \right)\\
&\,&\mathrm{exp} \left( \frac{\mu^2}{2} \,\|S^{\eta}\| \,
\left\{  (j+1)B + (k-(j+1)B)N\epsilon^2\right\}\right)  \left(  \frac{N^4\,r^{jB}}{\epsilon^2}\right).
\end{eqnarray}

Introduce the variable $\theta=\theta(j)=\frac{jB}{k}$. We now replace
the maximum over the discrete set $j=0,...,\mathcal{J}$, in eqn. \eqref{eqn-max} by
the supremum over $j \in [0,\mathcal{J}]$, i.e., $\theta \in [0,1]$; we get:
\begin{eqnarray}
\delta(k\mu) &\leq& \left( \mathcal{J}+1 \right) \sup_{\theta \in [0,1]}
\mathrm{exp} \left( |2 \mu^2-\mu|\, \overline{m} \, \left\{ 2 k \theta (1-N\sqrt{N}\epsilon)
+ B(2-N\sqrt{N}\epsilon)+kN\sqrt{N}\epsilon\right\} \right)\\
&\,&\mathrm{exp} \left( \frac{\mu^2}{2} \,\|S^{\eta}\| \,
\left\{  k \theta (1-N\epsilon^2)+B(1-N\epsilon^2)+k\epsilon^2\right\}\right)  \left(  \frac{N^4r^{jB}}{\epsilon^2}\right). \nonumber
\end{eqnarray}

Taking the limsup as $k \rightarrow \infty$, and then letting $\epsilon \rightarrow 0$, we obtain the following inequality:
 \begin{eqnarray}
 \label{eqn-sup}
 \limsup_{k\rightarrow \infty} \frac{1}{k} \log \mathcal{C}(k \mu) &\leq&
 \sup_{\theta \in [0,1]} \phi(\theta;\,\mu), \,\,\forall \mu>0\\
 \label{eqn-phi-fcn}
 \phi(\theta;\,\mu) &=&
 \frac{1}{2N^2} \sigma_L^2 \mu^2 + \frac{1}{N} m_L^{(0)}\,\mu +  2 |2\mu^2-\mu| \overline{m} \theta + \frac{\mu^2}{2} \|S^{\eta}\|  \theta + \theta \log r.
 \end{eqnarray}

Eqn. \eqref{eqn-sup} gives a family of upper bounds, indexed by $\mu>0$, on the
quantity $\limsup_{k\rightarrow \infty} \frac{1}{k} \log \alpha_{i,\mathrm{dis}}(k)$; we seek
 the most aggressive bound, i.e., we take the infimum over $\mu>0$.

 %Notice that the modulus function in eqn. \eqref{eqn-phi-fcn} has zero at $\mu=\frac{1}{2}$.
 We first discuss the values of the parameter $r$, for which there exists
 a range $\mu \in [\frac{1}{2}, \overline{\mu}]$ (where $\overline{\mu}$ is given by eqn. \eqref{eqn-mu-overline},) where
  the supremum in eqn.~\eqref{eqn-sup} is attained at $\theta=0$. It can be shown that this range is nonempty if and only if:
\begin{equation}
\label{eqn-cond-sada}
|\log r| > \frac{1}{8} \|S^{\eta}\|.
\end{equation}
In view of eqns. \eqref{eqn-sup}, \eqref{eqn-phi-fcn}, and \eqref{eqn-cond-sada}, we have the following inequality:
\begin{eqnarray}
\limsup_{k\rightarrow \infty} \frac{1}{k} \log \mathcal{C}^\prime(k \mu)
\leq \min_{\mu \in [\frac{1}{2},\overline{\mu}]} \left\{\frac{1}{2N^2} \sigma_L^2 \mu^2 + \frac{1}{N} m_L^{(0)}\,\mu\right\}.
\end{eqnarray}
The global minimum (on $\mu \in {\mathbb R}$) of the function
 $\mu \mapsto \frac{1}{2N^2} \sigma_L^2 \mu^2 + \frac{1}{N} m_L^{(0)}\,\mu$
  is attained at $\mu^{\star}=N/2$; and it can be shown that it equals the negative of the Chernoff information
   $-{\bf{C_{\mathrm{tot}}}}=-\frac{1}{8} (m_1-m_0)^\top S^{-1} (m_1-m_0).$
   Thus, a sufficient condition for
   $\limsup_{k \rightarrow \infty}\frac{1}{k} \log \alpha_{i,\mathrm{dis}}(k) \leq \limsup_{k\rightarrow \infty} \frac{1}{k} \log \mathcal{C}^\prime(k \mu)
   \leq -{\bf{C_{\mathrm{tot}}}}$ is that $\mu^{\star}=\frac{N}{2} \in [\frac{1}{2}, \overline{\mu}]$. By straightforward
   algebra, it can be shown that the latter condition translates into the following condition:
    \begin{eqnarray}
    \label{eqn_suf_cond}
    {|\log r|} \geq \frac{1}{8} N^2  \left( 1+(1-\frac{1}{N})K\right)\|S^{\eta}\| .
    \end{eqnarray}
The analysis of the upper bound on $\alpha_{i,\mathrm{dis}}(k)$
  remains true for the upper bound on $\beta_{i,\mathrm{dis}}(k)$; hence,
   we conclude that, under condition~\eqref{eqn_suf_cond}, we have:
   $\limsup_{k \rightarrow \infty}\frac{1}{k} \log \beta_{i,\mathrm{dis}}(k) \leq
    -{\bf{C_{\mathrm{tot}}}}$. Thus, under the condition~\eqref{eqn_suf_cond}, the following holds:
      \begin{eqnarray}
      \label{eqn-ovaovde}
      \limsup_{k \rightarrow \infty} \frac{1}{k} \log P^e_{i,\mathrm{dis}}(k) &=&\max
      \left( \limsup_{k \rightarrow \infty}\frac{1}{k} \log \alpha_{i,\mathrm{dis}}(k),\,\limsup_{k \rightarrow \infty}\frac{1}{k} \log \beta_{i,\mathrm{dis}}(k)  \right) \leq -{\bf{C_{\mathrm{tot}}}}.
      \end{eqnarray}
      On the other hand, by the Chernoff Lemma (Lemma \ref{lemma-Chernoff-gauss},) we also know that
      \begin{equation}
      \label{eqn-ovaovde2}
      \liminf_{k \rightarrow \infty} \frac{1}{k} \log P^e_{i,\mathrm{dis}}(k) \geq -{\bf{C_{\mathrm{tot}}}}.
      \end{equation}
      By eqns. \eqref{eqn-ovaovde} and \eqref{eqn-ovaovde2}, we conclude that
       $\lim_{k \rightarrow \infty} \frac{1}{k} \log P^e_{i,\mathrm{dis}}(k) = - {\bf{C_{\mathrm{tot}}}}$,
       for the values of $r$ that satisfy the condition \eqref{eqn_suf_cond}. Hence, condition \eqref{eqn_suf_cond} is a sufficient condition for asymptotic optimality.

       Consider now the values of $r$ such that:
       \begin{equation}
       \label{eqn-tata}
      \frac{1}{8} \|S^{\eta}\|<{|\log r|} < \frac{1}{8} N^2  \left( 1+(1-\frac{1}{N})K\right)\|S^{\eta}\|.
       \end{equation}
      In this range, $\frac{1}{2} < \overline{\mu}<\mu^\star=N/2$, and the minimum of the function $\mu \mapsto \frac{1}{2N^2} \sigma_L^2 \mu^2 + \frac{1}{N} m_L^{(0)}\,\mu$ on $\mu \in [\frac{1}{2}, \overline{\mu}]$ is attained at $\mu=\overline{\mu}$. Hence,
      for the values of $r$ in the range \eqref{eqn-tata}, we have the following bound:
      \begin{eqnarray}
      \limsup_{k \rightarrow \infty} \frac{1}{k} \log P^e_{i,\mathrm{dis}}(k) &=&\max
      \left( \limsup_{k \rightarrow \infty}\frac{1}{k} \log \alpha_{i,\mathrm{dis}}(k),\,\limsup_{k \rightarrow \infty}\frac{1}{k} \log \beta_{i,\mathrm{dis}}(k)  \right) \\
      &\leq& \frac{1}{2N^2} \sigma_L^2 \overline{\mu}^2 + \frac{1}{N} m_L^{(0)}\,\overline{\mu}.
      \end{eqnarray}

Finally, consider the values of $r$ such that the condition in \eqref{eqn-cond-sada} does not hold, i.e.,
$\frac{1}{8} \|S^{\eta}\| \geq {|\log r|}$. We omit further details, but it can be shown that, in this case,
the supremum in eqn.~\eqref{eqn-sup} is attained at $\theta=0$ only for
the values of $\mu$ in a subset of $[0, \frac{1}{2}]$; and it can be shown that $\limsup_{k \rightarrow \infty} \frac{1}{k} \log P^e_{i,\mathrm{dis}}(k)$ can be bounded as given by eqn. \eqref{eqn-rate-bound}.
\end{proof}

\subsection{Necessary condition for asymptotic optimality}

We proceed with necessary conditions for asymptotic optimality, for the case when the sensors' observations
are spatially uncorrelated, i.e., the noise covariance $S=\mathrm{Diag} \left(\sigma_1^2,...,\sigma_N^2 \right)$.
Denote by $E_i(k)$ the event
 that, at time $k$,
  sensor $i$ is connected to at least
  one of the remaining sensors in
   the network; that is,
   \begin{equation}
   E_i(k):= \left\{ \max_{j=1,...,N,\,j\neq i} W_{ij}(k) >0 \right\}.
   \end{equation}
   Further, denote by $P_i(k)=P_i=\mathbb{P}(E_i(k))$
     We have the following result.

\begin{theorem}[Necessary condition for asymptotic optimality]
\label{theorem-necess-cond-general}
Consider the distributed detection algorithm in eqns. \eqref{eqn-running-cons-sensor-i} and \eqref{eqn-decision-x-i} with
 spatially uncorrelated sensors' observations, and let Assumption~\ref{assumption-W(k)} hold.
 Then, a necessary condition for the asymptotic optimality of distributed detection at sensor $i$ is:
  \begin{equation}
  |\log(1-P_i)| > {\bf{C_{\mathrm{tot}}}} -  {\bf{C}_i}.
  \end{equation}
  \end{theorem}

  \begin{proof} Consider $\alpha_{i,\mathrm{dis}}(k)$,
   but remark that the same conclusions will
   hold for $\beta_{i,\mathrm{dis}}(k)$, as
   $\alpha_{i,\mathrm{dis}}(k)=\beta_{i,\mathrm{dis}}(k).$ (See the proof of Theorem \ref{theorem-neces-suf-special-case}.)
    Further, using inequalities \eqref{eqn-Q-ineq}, we have:
    \begin{eqnarray}
    \alpha_{i,\mathrm{dis}}(k) &=& \mathbb{P}_0 \left(  x_i(k)>0\right) \\
    &\geq& \mathbb{P}_0 \left(  x_i(k)>0| \cap_{s=1}^k E_i(s)\right)
     \mathbb{P} \left( \cap_{s=1}^k E_i(s) \right) \\
     &\geq& \frac{1}{\sqrt{2 \pi}}\, \frac{\sqrt{2 k} {\bf C}_i } {1+2k{\bf C}_i} e^{-k {\bf{C}_i}} (1-P_i)^k.
    \end{eqnarray}
    Thus, we have that
    \begin{equation}
    \label{eqn-neces-1}
    \liminf_{k \rightarrow \infty} \frac{1}{k} \log \alpha_{i,\mathrm{dis}}(k) \geq -
    \left( {\bf{C}_i} + |\log(1-P_i)| \right).
    \end{equation}
   On the other hand, we know that
   \begin{eqnarray}
   \label{eqn-neces-2}
   \liminf_{k \rightarrow \infty} \frac{1}{k} \log P^e_{i,\mathrm{dis}}(k) &\geq& \liminf_{k \rightarrow \infty} \frac{1}{k} \log \alpha_{i,\mathrm{dis}}(k)\geq -{ \bf{C}_{\mathrm{tot}}  }.
   \end{eqnarray}
   The claim of Theorem~\ref{theorem-necess-cond-general} now follows from eqns.~\eqref{eqn-neces-1}~and~\eqref{eqn-neces-2}.
   \end{proof}

\section{Simulations}
\label{section-simul}
In this section, we corroborate by simulation examples our analytical findings
on the asymptotic behavior of distributed detection over random networks. Namely,
we demonstrate the ``phase change'' behavior of distributed detection with respect to
the speed of network information flow, as predicted by Theorem \ref{theorem-sufficient-condition}. Also,
 we demonstrate that a sensor with poor connectedness to the rest of the network cannot be an optimal
 detector, as predicted by Theorem \ref{theorem-necess-cond-general}; moreover,
 its performance approaches the performance
 of an isolated sensor, i.e., a sensor that works as an individual detector, as connectedness becomes worse and worse.

\mypar{Simulation setup} We consider a supergraph $G$ with $N=40$
 nodes and $M=247$ edges.  Nodes are uniformly
 distributed on a unit square and nodes within distance less than a radius $r$
  are connected by an edge. %We model
%   the formation probability of a link $\{i,j\} \in E$ (the probability of link being online at a time step) as follows:
%   \[
%   \pi_{ij} = 1 - c\, \frac{\Delta_{ij}^2}{r^2}, \,\,c=0.6.
%   \]
%   Thus, the formation probability of a link $\{i,j\} \in E$ decreases quadratically with the distance
%    $\Delta_{ij}$ between sensors $i$ and $j$, and the minimal link formation probability is $0.4$.
As averaging weights,
    we use the standard time-varying Metropolis weights $W_{ij}(k)$, defined for
    $\{i,j\} \in E, i\neq j$, by
     $W_{ij}(k) = 1/(1+\max(d_i(k),d_j(k)))$, if the link $\{i,j\}$
      is online at time $k$, and 0 otherwise. The quantity
      $d_i(k)$ represents the
      number of neighbors (i.e., the degree) of node $i$
      at time $k$. Also, $W_{ii}(k) = 1-\sum_{j \in \Omega_i(k)}W_{ij}(k)$, for all $i$,
       and $W_{ij}(k) \equiv 0$, for $i \neq j$, $\{i,j\} \notin E$. The link failures are spatially and temporally independent.
       Each link $\{i,j\}\in E$ has the same probability of formation, i.e., the probability of being online at a time, $q_{ij}=q$. This network and weight model satisfy Assumption~\ref{assumption-W(k)}.

We assume equal prior probabilities, $\mathbb{P}(H_0)=\mathbb{P}(H_1)=0.5$. We set the $N\times 1$
 signal vector under $H_1$ (respectively, $H_0$) to be $m_1=1$ (respectively, $m_0=0$.)
 We generate randomly the covariance matrix $S$, as follows. We generate: a $N \times N$ matrix $M_S$, with
  the entries drawn independently from $U[0,1]$--the uniform distribution on $[0,1]$; we set $R_S=M_S M_S^\top$; we decompose $R_S$ via the eigenvalue
  decomposition: $R_S = Q_S \Lambda_S Q_S^\top$; we generate a $N \times 1$
  vector $u_S$ with the entries drawn independently from $U[0,1]$; finally, we set
  $S = \alpha_S\,Q_S \mathrm{Diag}(u_S) Q_S^\top$, where $\alpha_S>0$
   is a parameter. For the optimal centralized detector, we evaluate
$P^e_{\mathrm{cen}}(k)$ by formula~\eqref{eqn-theorem-chernoff-lemma}.
For the distributed detector, we
evaluate $P^e_{i,\mathrm{dis}}(k)$ by Monte Carlo simulations
 with 20,000 sample paths (20,000 for each hypothesis $H_l$, $l=0,1$)
 of the running consensus algorithm.

\mypar{Exponential rate of decay of the error probability vs. the speed of information flow} First, we examine the asymptotic behavior
of distributed detection when the speed of network information flow varies, i.e., when $r$ varies.
(We recall that $r:=\lambda_2\left( \mathbb{E} \left[ W(k)^2 \right] \right)$.) To this end,
 we fix the supergraph, and then we vary the formation probability of links $q$ from 0 to 0.75. Figure \ref{Figure_phase change} (bottom left) plots the estimated exponential rate of decay, averaged across sensors, versus
    $q$. Figure \ref{Figure_phase change} (bottom right) plots the same estimated exponential rate of decay versus $1-r$. We can see that there is a ``phase change'' behavior, as predicted by Theorem \ref{theorem-sufficient-condition}.
     For $q$ greater than 0.1, i.e.,
     for $1-r > 0.25$, the rate of decay of error probability is approximately the same
     as for the optimal centralized detector ${\bf{C_{\mathrm{tot}}}}$--the simulation estimate of ${\bf{C_{\mathrm{tot}}}}$ is 0.0106.
     \footnote{In this
     numerical example, the theoretical value of ${\bf{C_{\mathrm{tot}}}}$ is 0.009. The estimated value
     shows an error because the decay of the error probability, for the centralized detection,
      and for distributed detection with a large $q$, tends to slow down slightly when $k$ is very large; this effect is not
      completely captured by simulation with $k < 700.$} For $q<0.1$, i.e., for $1-r<0.25$, detection performance becomes worse and worse as $q$ decreases.
     Figure \ref{Figure_phase change} (top) plots the estimated error probability, averaged across sensors,
for different values of $q$. We can see that the curves are ``stretched''
for small values of $q$; after $q$ exceeds a threshold (on the order of $0.1$,)
 the curves cluster, and they have approximately the same slope (the error probability has approximately the same decay rate,) equal to the optimal slope.

\mypar{Study of a sensor with poor connectivity to the rest of the network}
Next, we demonstrate that a sensor with poor connectivity to the rest of the network
cannot be an asymptotically optimal detector, as predicted by Theorem \ref{theorem-necess-cond-general}; its
performance approaches the performance of an individual detector-sensor, when its connectivity
  becomes worse and worse. For $i$-th individual detector-sensor (no cooperation between sensors,)
   it is easy to show that the
   Bayes probability of error, $P^e_{i,\mathrm{no\,\,cooper.}}(k)$ equals:
   $
   P^e_{i,\mathrm{no\,\,cooper.}}(k) = \mathcal{Q} \left(  \sqrt{k} \frac{m_{i,\mathrm{no\,\,cooper.}}}{\sigma_{\mathrm{i,no\,\,cooper.}}}\right),
   $
   where $m_{i,\mathrm{no\,\,cooper.}}=\frac{1}{2}\frac{[m_1]_i^2}{S_{ii}}$, and
   $\sigma_{\mathrm{i,no\,\,cooper.}}^2 = \frac{[m_1]_i^2}{S_{ii}}.$
   It is easy to show that the Chernoff information
   (equal to $\lim_{k \rightarrow \infty}\frac{1}{k} \log P^e_{i,\mathrm{no\,\,cooper.}}(k)$) for sensor $i$,
    in the absence of cooperation, is given by $\frac{1}{8}\frac{[m_1]_i^2}{S_{ii}}$.

We now detail the simulation setup. We consider a supergraph with $N=35$ nodes and $M=263$ edges. We initially generate the supergraph
as a geometric disc graph, but then we isolate sensor 35 from the rest of the network,
by keeping it connected only to sensor 3. We then vary the formation probability
 of the link $\{3,35\}$, $q_{3,35}$, from 0.05 to 0.5 (see Figure \ref{Figure_isolated_sensor}.) All other
 links in the supergraph have the formation probability of 0.8.
 Figure \ref{Figure_isolated_sensor} plots
 the error probability for: 1) the optimal centralized detection; 2)
 distributed detection at each sensor, with cooperation (running consensus;) and 3)
 detection at each sensor, without cooperation (sensors do not communicate.) Figure \ref{Figure_isolated_sensor} shows that, when $q_{3,35}=0.05$, sensor 35
 behaves almost as bad as the individual sensors, that do not communicate (cooperate) with each other. As $q$ increases, the performance of sensor 35 gradually improves.

\begin{figure}[thpb]
      \centering
      \includegraphics[height=2.7in,width=3.99in]{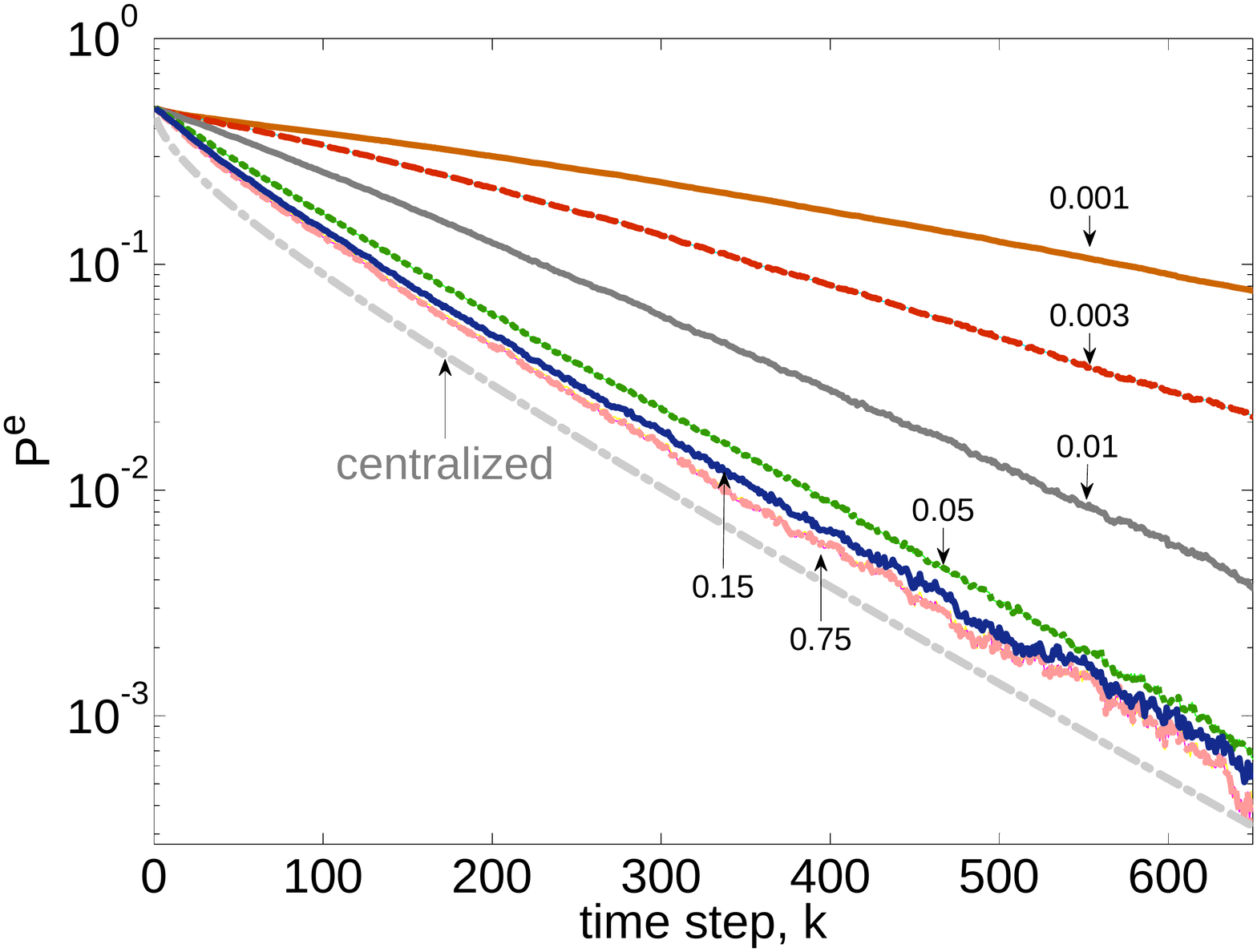}
      \includegraphics[height=2.in,width=3.1in]{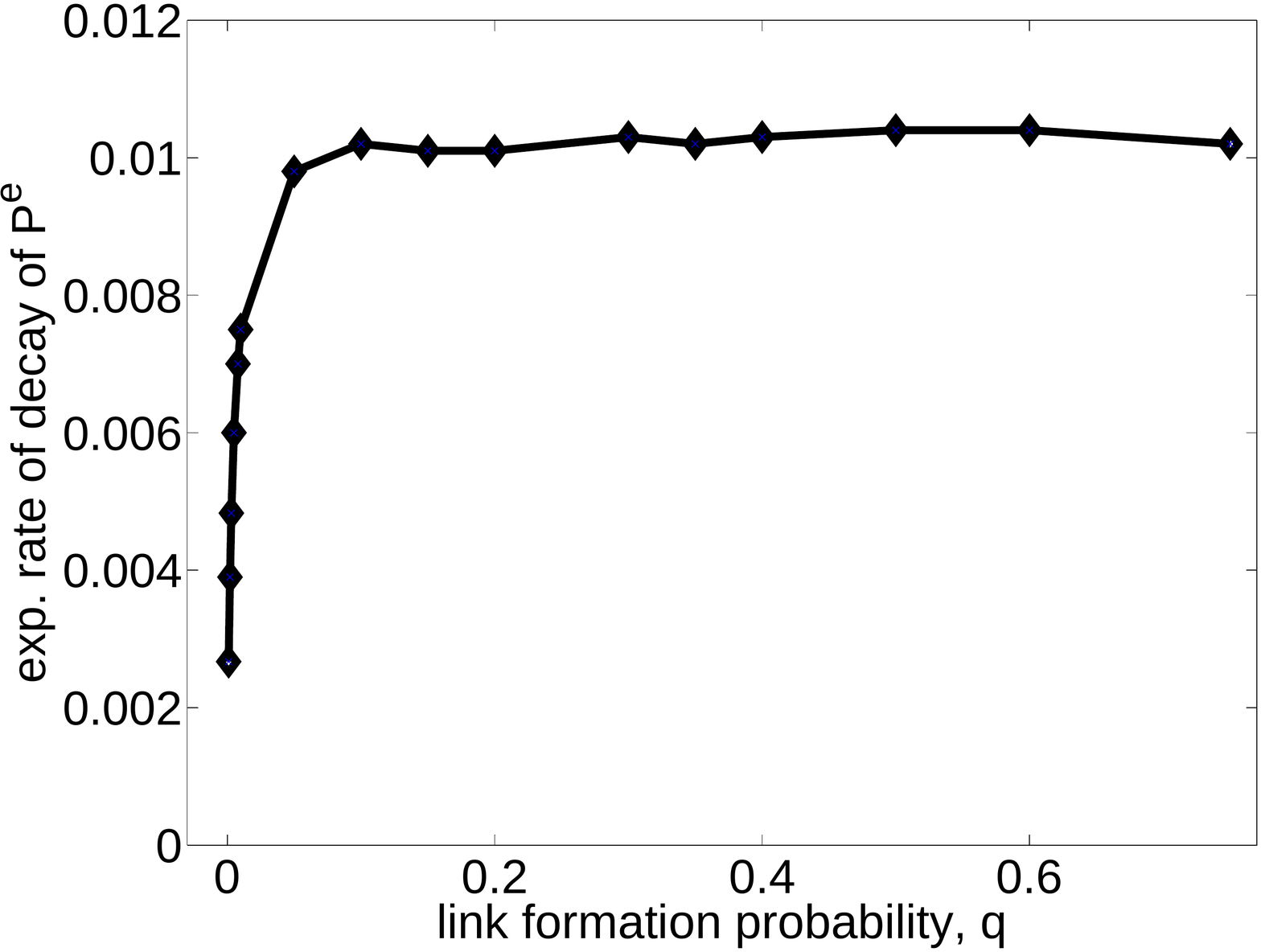}
      \includegraphics[height=2.in,width=3.1in]{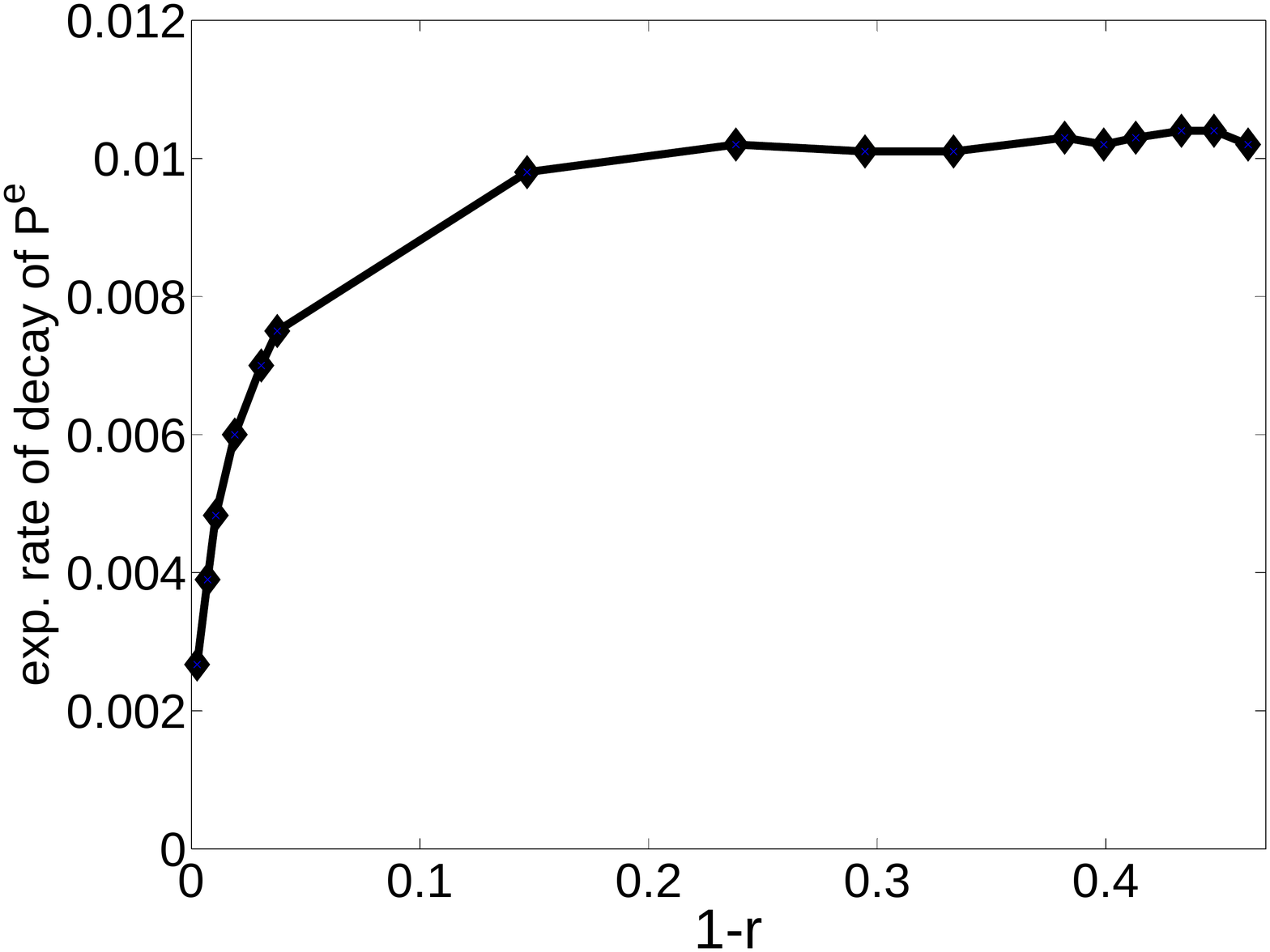}
      \caption{ Monte Carlo estimate of the performance of distributed detection for different values
      of the link formation probability $q$. Top: Error probability averaged across $N$ sensors.
      Each line is labeled with the value of $q$; performance of centralized detection is plotted in gray. Bottom left (respectively, right): Estimated exponential rate of decay of the error probability vs. $q$ (respectively, $1-r$). Recall that
      $r:=\lambda_2\left( \mathbb{E} \left[ W(k)^2 \right] \right)$.}
      \label{Figure_phase change}
\end{figure}

\begin{figure}[thpb]
      \centering
      \includegraphics[height=2.in,width=3.1in]{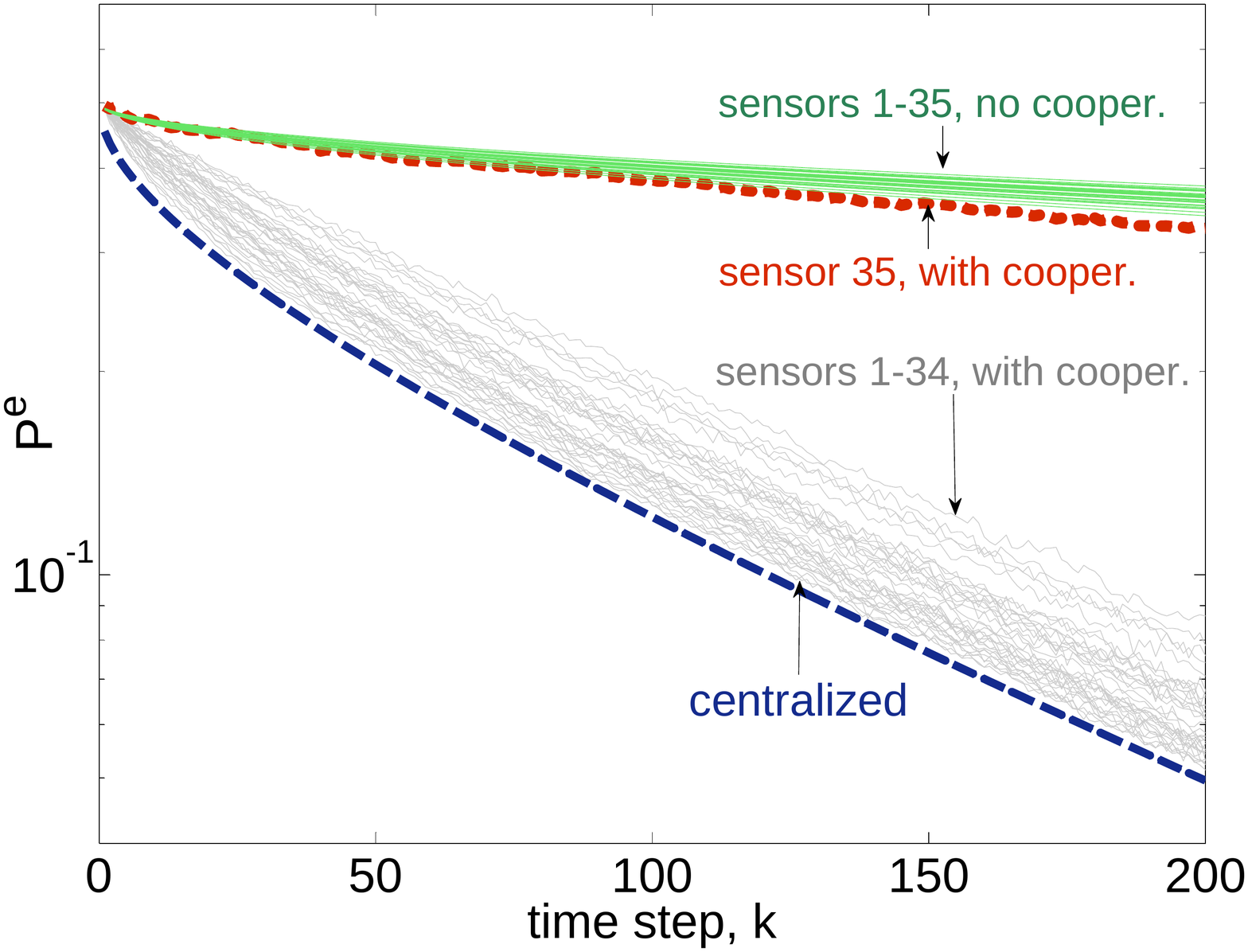}
      \includegraphics[height=2.in,width=3.1in]{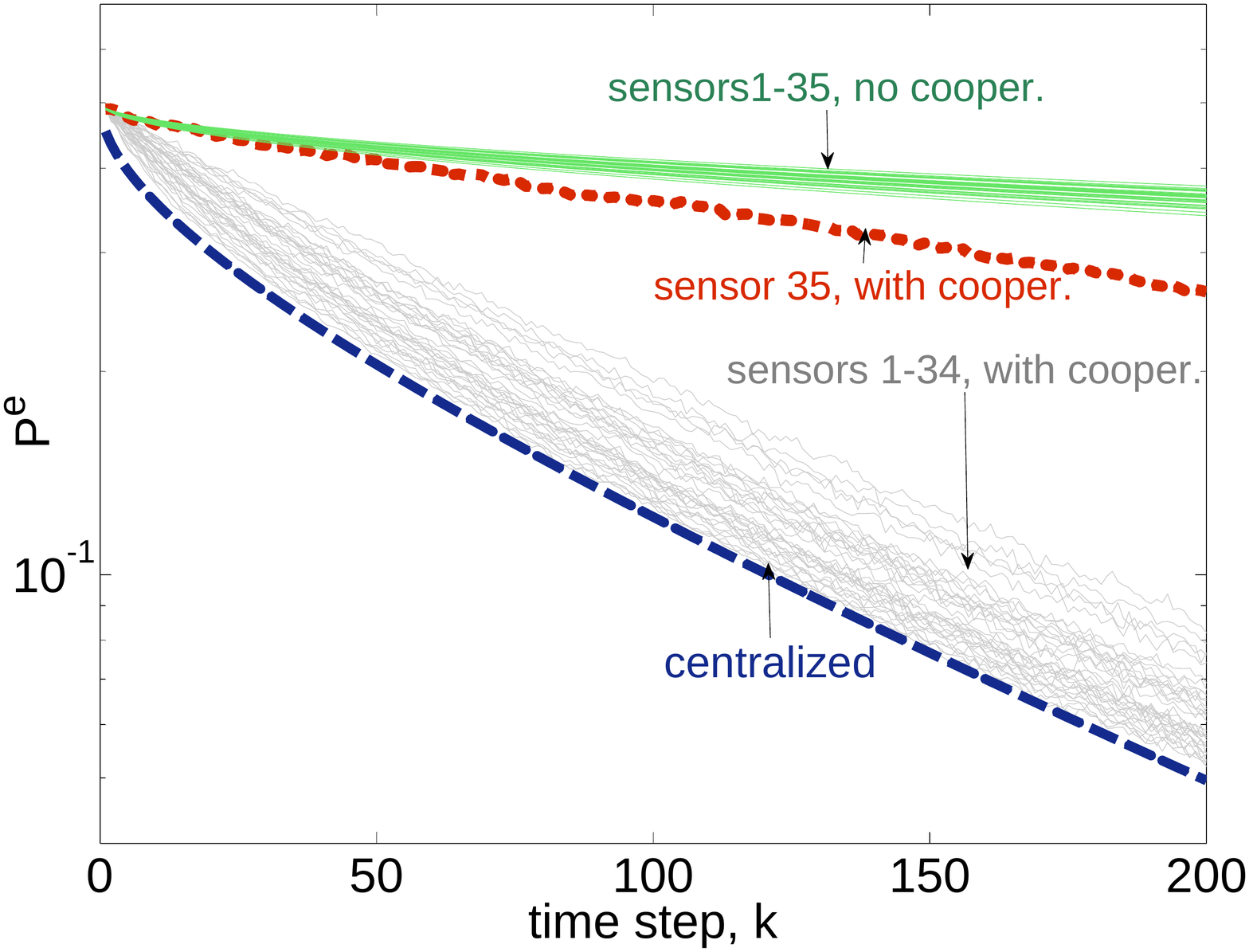}
      \includegraphics[height=2.in,width=3.1in]{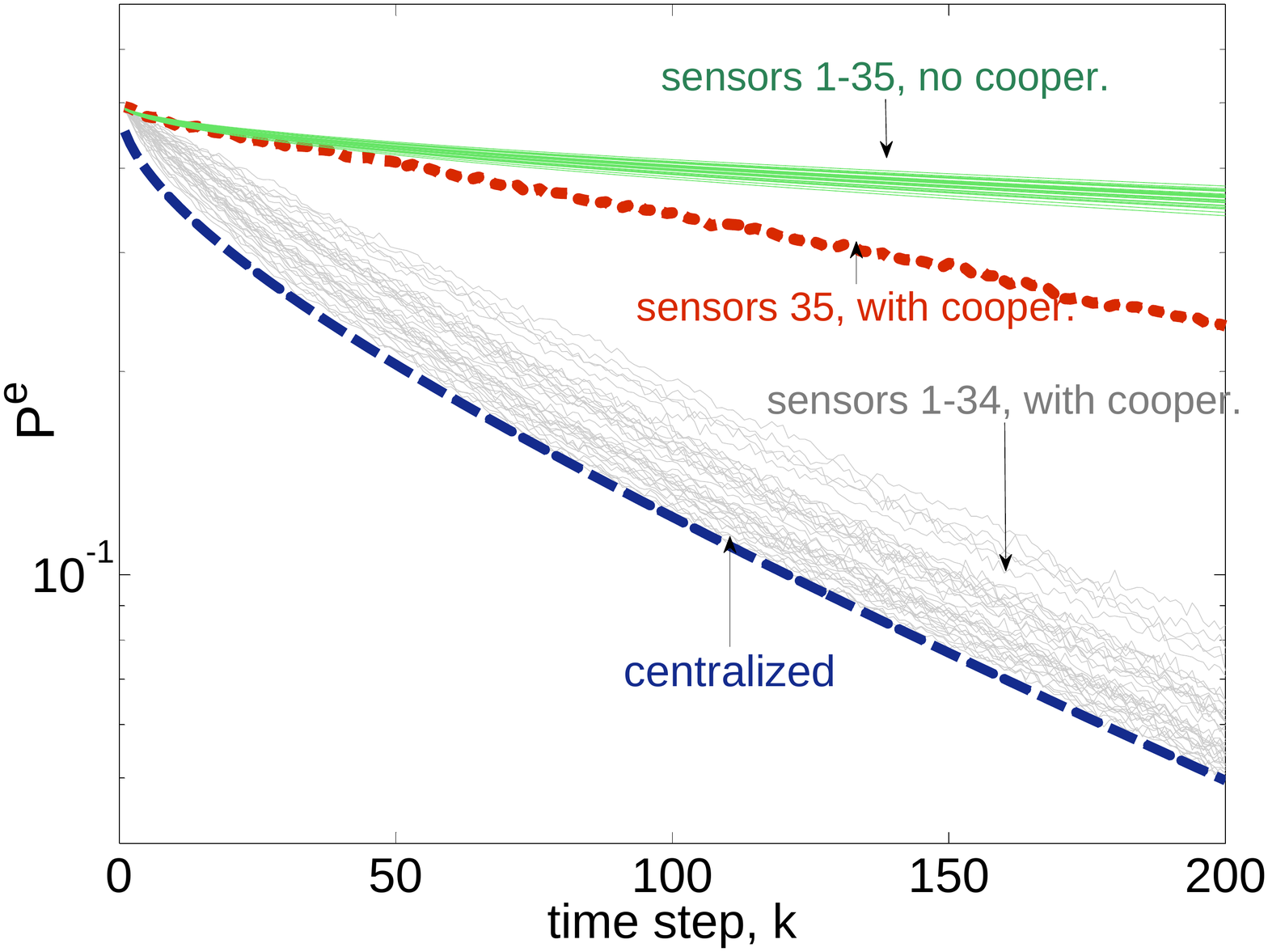}
      \includegraphics[height=2.in,width=3.1in]{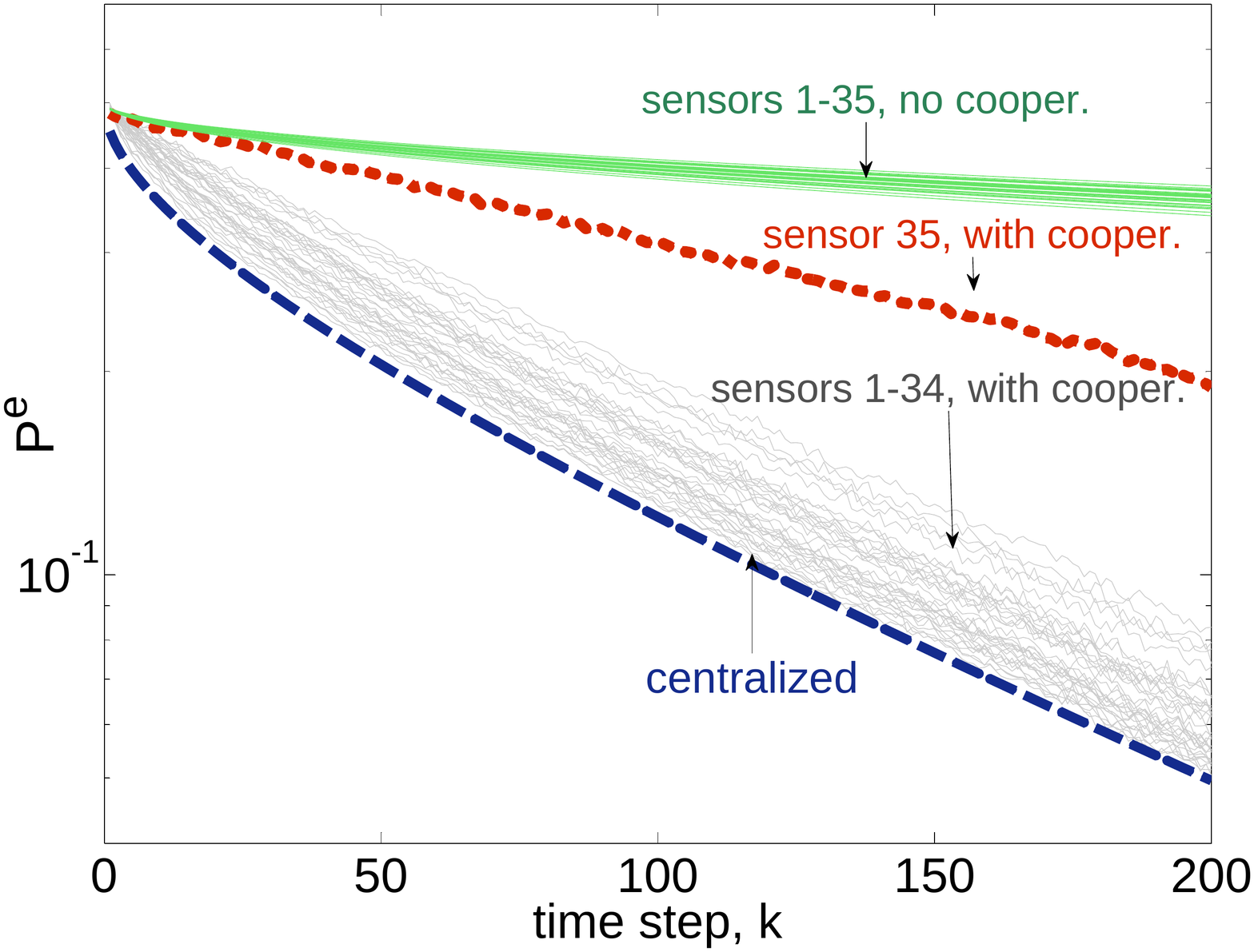}
      \caption{ Error probability averaged across sensors for the optimal centralized detection,
      distributed detection at each sensor (with cooperation), and detection at each sensor, without cooperation.
      The formation probability $q_{3,35}$ of the link $\{3,35\}$ varies between 0.05 and 0.5:
      $q_{3,35}$=0.05 (top right); 0.2 (top left); 0.3 (bottom left); 0.5 (bottom right).}
      \label{Figure_isolated_sensor}
\end{figure}

\section{Conclusion}
We studied distributed detection over random networks when
at each time step $k$ each sensor: 1) averages its decision variable
with the neighbors' decision variables; and 2) accounts
on-the-fly for its new observation. We analyzed
how asymptotic detection performance, i.e., the exponential decay rate of the error
probability, depends on the random network connectivity, i.e., on the speed of information flow across network.
 We showed that distributed detection exhibits a ``phase change.'' Namely, distributed detection is asymptotically optimal,
  if the network speed of information flow is above a Chernoff information dependent threshold.
  When below the threshold, we find a lower bound on the achievable performance (the exponential rate of decay of error probability,) as a function of the network connectivity and the Chernoff information. Simulation examples
  demonstrate our theoretical findings on the asymptotic performance of distributed detection.
\vspace{-20mm}
\appendix

\section{Appendix}

\subsection{Proof of inequalities \eqref{eqn-hat-phi-limsup}}
Recall $\chi(l;k)$ in eqn. \eqref{eqn-def-chi}. Using eqn. \eqref{eqn-alpha-total-prob-law},
and noting that $(1-p)^{k-1} \geq p(1-p)^{k-1}$, we bound $\alpha_{i,\mathrm{dis}}(k)$
 from below as follows:
\begin{eqnarray}
\label{eqn-inequalities}
\alpha_{i,\mathrm{dis}}(k) &\geq& \max_{l=0,...,k-1} \,\mathcal{Q} (\chi(l;k)) \, p(1-p)^{k-l-1} \\
\label{eqn-ineq-ova}
&\geq& \frac{1}{\sqrt{2 \pi }}\,\max_{l=0,...,k-1} \, \frac{\chi(l;k)}{1+\chi(l;k)^2} e^{-\frac{1}{2}\chi(l;k)^2} \,p(1-p)^{k-l-1}\\
&\geq& \frac{1}{\sqrt{2 \pi }}\,\frac{\sqrt{2k{\bf{C_i}}}}{1+2kN{\bf{C_i}}} \max_{l=0,...,k-1} e^{-\frac{1}{2}\chi(l;k)^2} \,p(1-p)^{k-l-1}\\
&=&    \frac{1}{\sqrt{2 \pi }}\,\frac{\sqrt{2k{\bf{C_i}}}}{1+2kN{\bf{C_i}}}
       \max_{l=0,...,k-1}  p \,e^{-\frac{k {\bf{C_{\mathrm{tot}}}}} {1+(N-1)(1-\frac{l}{k})} - (k-l-1)|\log(1-p)|}\\
&=&
\label{eqn-equal-ova}
\frac{1}{\sqrt{2 \pi }}\,\frac{\sqrt{2k{\bf{C_i}}}}{1+2kN{\bf{C_i}}}
       \max_{j=0,...,k-1}  p \,e^{-\frac{k {\bf{C_{\mathrm{tot}}}}} {1+(N-1)\frac{(j+1)}{k}} - j |\log(1-p)|}\\
       &=&
\label{eqn-equal-phi-last}
\frac{1}{\sqrt{2 \pi }}\,\frac{\sqrt{2k{\bf{C_i}}}}{1+2kN{\bf{C_i}}}
       \max_{j=0,...,k-1}  p \,e^{-k\,\phi(j;k)}.
\end{eqnarray}
Inequality in eqn.~\eqref{eqn-ineq-ova} is due to~\eqref{eqn-Q-ineq}, equality
 in~\eqref{eqn-equal-ova} is by letting $j=k-l-1,$ and equality
 in~\eqref{eqn-equal-phi-last} is by definition of $\phi(j;k)$ in eqn.~\eqref{eqn-def-phi-j-k}. Taking $\liminf_{k \rightarrow \infty} \frac{1}{k} \log \alpha_{i,\mathrm{dis}}(k)$, and using \eqref{eqn-liminf-p-e}, \eqref{eqn-inequalities}-\eqref{eqn-equal-phi-last} yields:
\begin{eqnarray}
\liminf_{k \rightarrow \infty} \frac{1}{k} \log P^e_{i,\mathrm{dis}}(k) &\geq& \liminf_{k \rightarrow \infty} \frac{1}{k} \log \alpha_{i,\mathrm{dis}}(k)
\geq
 - \limsup_{k \rightarrow \infty} \left\{  \min_{j=0,...,k-1} \phi(j;k) \right\} .
\end{eqnarray}

Next, using eqn. \eqref{eqn-alpha-total-prob-law}, we bound $\mathbb{P}_0 \left( x_i(k)>0  \right)$
 from above as follows:
\begin{eqnarray}
\label{eqn-inequalities-2}
\alpha_{i,\mathrm{dis}}(k) &\leq& k\,\max_{l=0,...,k-1} \,\mathcal{Q} (\chi(l;k)) \, (1-p)^{k-l-1} \\
\label{eqn-ineq-ova-2}
&\leq& k\,\frac{1}{\sqrt{2 \pi}}\,  \max_{l=0,...,k-1} \, \frac{1}{\chi(l;k)} e^{-\frac{1}{2}\chi(l;k)^2} \,(1-p)^{k-l-1}\\
&\leq& k\,\frac{1}{\sqrt{2 \pi}}\,  \frac{1}{\sqrt{2{\bf{C_i}}k}} \max_{l=0,...,k-1} e^{-\frac{1}{2}\chi(l;k)^2} \,(1-p)^{k-l-1}\\
&=&    k\,\frac{1}{\sqrt{2 \pi}}\,  \frac{1}{\sqrt{2{\bf{C_i}}k}}\,
       \max_{l=0,...,k-1} \,  \,e^{-\frac{k {\bf{C_{\mathrm{tot}}}}} {1+(N-1)(1-\frac{l}{k})} - (k-l-1)|\log(1-p)|}\\
       &=&
\label{eqn-equal-ova-2}
k\, \frac{1}{\sqrt{2 \pi}}\, \frac{1}{\sqrt{2{\bf{C_i}}k}}\,
       \max_{j=0,...,k-1} \,e^{-\frac{k {\bf{C_{\mathrm{tot}}}}} {1+(N-1)\frac{j}{k}} - j |\log(1-p)|}\\
       &=&
\label{eqn-zadnja-def-phi}
k\, \frac{1}{\sqrt{2 \pi}}\, \frac{1}{\sqrt{2{\bf{C_i}}k}}\,
       \max_{j=0,...,k-1} \,e^{-k \,\phi(j;k)}.
\end{eqnarray}
Inequality in~\eqref{eqn-ineq-ova-2} is due~\eqref{eqn-Q-ineq}, and equality
 in~\eqref{eqn-zadnja-def-phi} is by definition of $\phi(j;k)$ in eqn.~\eqref{eqn-def-phi-j-k}. We complete the proof by noting that
 equality in \eqref{eqn-limsup-p-e} and inequalities \eqref{eqn-inequalities-2}-\eqref{eqn-equal-ova-2} imply that:
\begin{equation}
\limsup_{k \rightarrow \infty} \frac{1}{k} \log P^e_{i,\mathrm{dis}}(k) = \limsup_{k \rightarrow \infty} \frac{1}{k} \log \alpha_{i,\mathrm{dis}}(k) \leq
 - \liminf_{k \rightarrow \infty} \left\{  \min_{j \in \{0,...,k-1\}} \phi(j;k) \right\}.
\end{equation}
\vspace{-5mm}
\subsection{Proof of Lemma \ref{lemma-phi-epsilon}}
Consider the random vector $z^{(i)}(k,j)=\widetilde{\Phi}(k,j)e_i$, which is equal to the $i$-th
 column of the matrix $\widetilde{\Phi}(k,j)$. First, by the Markov inequality, we have:
 \begin{equation}
 \label{eqn-markov-ineq}
 \mathbb{P} \left(  \| z^{(i)}(k,j) \|  > \epsilon \right)
 \leq
\frac{1}{\epsilon^2} \mathbb{E} \left[ z^{(i)}(k,j)^\top z^{(i)}(k,j)\right].
 \end{equation}
 Now, we can bound $\mathbb{E} \left[ z^{(i)}(k,j)^\top z^{(i)}(k,j)\right]$
  as follows:
 \begin{eqnarray}
 \label{eqn-bounding}
 \mathbb{E} \left[ z^{(i)}(k,j)^\top z^{(i)}(k,j)\right]
 &=&  \mathbb{E} \left[ \mathbb{E} \left[ e_i^\top \widetilde{\Phi}(k-1,j)^\top\widetilde{W}(k-1)\widetilde{W}(k-1)
 \widetilde{\Phi}(k-1,j)e_i\,|\widetilde{W}(j),...,\widetilde{W}(k-2)\right]  \right]\nonumber\\
 &=&   \mathbb{E} \left[e_i^\top \widetilde{\Phi}(k-1,j)^\top\mathbb{E}\left[\widetilde{W}(k-1)\widetilde{W}(k-1)\right]\widetilde{\Phi}(k-1,j)
 e_i|\,\widetilde{W}(j),...,\widetilde{W}(k-2)\right]  \nonumber\\
 &\leq&  \|\mathbb{E}\left[\widetilde{W}(k-1)^2\right] \| \,\cdot\,
 \mathbb{E}\left[ e_i^\top \widetilde{\Phi}(k-1,j)^\top\widetilde{\Phi}(k-1,j)e_i\right].\nonumber\\
 &=& r \,
 \mathbb{E}\left[z^{(i)}(k-1,j)^\top z^{(i)}(k-1,j)\right].
 \end{eqnarray}
Now, bounding successively $\mathbb{E}\left[z^{(i)}(k-s,j)^\top z^{(i)}(k-s,j)\right]$ for
 $s=1,...,k-j$, as in eqn.~\eqref{eqn-bounding}, we obtain:
 \begin{equation}
  \mathbb{E} \left[ z^{(i)}(k,j)^\top z^{(i)}(k,j)\right] \leq r^{k-j},
 \end{equation}
 which yields (by eqn.~\eqref{eqn-markov-ineq}:)
 \begin{equation}
 \mathbb{P} \left(  \| z^{(i)}(k,j) \|  > \epsilon \right)
\leq
\frac{1}{\epsilon^2} r^{k-j}, \,i=1,...,N.
\end{equation}
Now, observe that $\|z^{(i)}(k,j)\|_1 \leq N \|z^{(i)}(k,j)\|$, and hence,
\[
\mathbb{P} \left(  \| z^{(i)}(k,j) \|_1  > \epsilon \right)
\leq \mathbb{P} \left(  \| z^{(i)}(k,j) \|  > \frac{\epsilon}{N} \right).
\]
Further, in view of the matrix norm inequality $\|B\| \leq \sqrt{N}\|B\|_1$, we have:
\begin{eqnarray}
\mathbb{P} \left(  \| \widetilde{\Phi}(k,j) \|  > \epsilon \right)
&\leq &  \mathbb{P} \left(  \| \widetilde{\Phi}(k,j) \|_1  > \frac{\epsilon}{\sqrt{N}} \right)=    \mathbb{P} \left(  \max_{i=1,...,N} \| z^{(i)} (k,j)\|_1  > \frac{\epsilon}{\sqrt{N}} \right) \nonumber
\\
&=&     \mathbb{P}  \left(  \cup_{i=1}^N \{  \| z^{(i)} (k,j)\|_1  > \frac{\epsilon}{\sqrt{N}}  \right)
 \leq
 \sum_{i=1}^N   \mathbb{P} \left( \| z^{(i)} (k,j)\|_1  > \frac{\epsilon}{\sqrt{N}} \right) \leq  \frac{N^4}{\epsilon^2} r^{k-j}. \nonumber
\end{eqnarray}

\vspace{-4mm}
\subsection{Proof of Lemma \ref{lemma-event-a-l}}

First, note that, if $\mathcal{A}_j$ occurred, we have that:
\[
\| \widetilde{\Phi}(k,s) \| \leq \epsilon, \,\forall s \leq k-jB,
\]
since $
\|\widetilde{\Phi}(k,s) \| = \| [\widetilde{\Phi}(k,k-jB) \widetilde{\Phi}(k-jB-1,s) \| \leq
\| \widetilde{\Phi}(k,k-jB) \| \| \widetilde{\Phi}(k-jB-1,s) \|,
$
and $\| [\widetilde{\Phi}(k,k-jB) \| \leq \epsilon$ and $\| \widetilde{\Phi}(k-jB-1,s) \| \leq 1.$ Inequality \eqref{eqn-prva} holds true, because:
\begin{eqnarray*}
\sum_{l=1}^N |[\widetilde{\Phi}(k,s)]_{il}| &\leq& N \, \max_{i,l=1,...,N} | [\widetilde{\Phi}(k,s)]_{il} |,
\end{eqnarray*}
and $
\max_{i,l=1,...,N} | [\widetilde{\Phi}(k,s)]_{il} | \leq \sqrt{N} \, \|\widetilde{\Phi}(k,s)\| \leq \sqrt{N} \epsilon.
$
Inequality \eqref{eqn-druga} holds true, because, for $s \leq k-jB$,
\[
\sum_{l=1}^N |[\widetilde{\Phi}(k,s)]_{il}|^2 \leq \|\widetilde{\Phi}(k,s)\|_F^2 \leq N \|\widetilde{\Phi}(k,s)\|^2 = N \epsilon^2.
\]

\subsection{Proof of eqn. \eqref{eqn-append}}

Consider the Chernoff bound on $\alpha_{i,\mathrm{dis}}(k)$, given by:
\begin{eqnarray}
\mathcal{C} \left( k\mu \right) &:=& {\mathbb{E}}_0 \left[  \mathrm{exp}\,\left(k \mu \,x_i(k)\right)
   \right]
=   {\mathbb E}_0 \left[  \mathrm{exp}\, \left( k \lambda^\top x(k) \right)
   \right] ,
\end{eqnarray}
where $\lambda=\mu\,e_i$, $\lambda \in {\mathbb R}^N$.

Further, we have:
\begin{eqnarray}
\label{eqn-term-to-drop}
\mathcal{C} \left( k\mu \right)
&=&
 {\mathbb E}_0
\left[
\mathrm{exp}\, \left(    \lambda^\top \sum_{j=1}^{k-1} {\Phi}(k, j) \eta(j)  + \, \lambda^\top \eta(k)      \right)
\right]
\\
&=&
  {\mathbb E}_0
\left[
\mathrm{exp}\, \left(    \lambda^\top \sum_{j=1}^{k-1} {\Phi}(k, j) \eta(j)        \right)
\right]\,  {\mathbb E}_0 \left[  \mathrm{exp} \left(  \lambda^\top \eta(k) \right)  \right]
, \nonumber
\end{eqnarray}
where the last equality holds because $\eta(k)$ is independent from
$\eta(j)$ and $W(j)$, $j=1,...,k-1$. We will be interested in computing $\limsup_{k \rightarrow \infty} \frac{1}{k} \Lambda_k^{(l)}(k\,\mu)$, for
all $\mu \in {\mathbb R}$; with this respect,
remark that
\[
\lim_{k \rightarrow \infty} \frac{1}{k} \log   {\mathbb E}_0 \left[  \mathrm{exp} \left(  \lambda^\top \eta(k) \right) \right] =0,
\]
for all $\lambda \in {\mathbb R}^N$, because $\eta(k)$ is a Gaussian random variable
and hence it has finite log-moment generating function at any point $\lambda$.

Thus, we have that $\limsup_{k \rightarrow \infty} \frac{1}{k} \log \mathcal{C} \left( k\mu \right)=
\limsup_{k \rightarrow \infty} \frac{1}{k} \log \mathcal{C}^\prime(k\mu)$, where
\[ \mathcal{C}^\prime(k\mu)=
  {\mathbb E}_0
\left[
\mathrm{exp}\, \left(  \,  \lambda^\top \sum_{j=1}^{k-1} {\Phi}(k, j) \eta(j)        \right) \right] .
\]
We thus proceed with the computation of $\mathcal{C}^\prime(k\mu)$.
Conditioned on $W(1),W(2),...,W(k-1)$,
 the random variables $ \lambda^\top \Phi(k,j) \eta(j)$, $j=1,...,k-1$,
  are independent; moreover, they are Gaussian random variables,
  as linear transformation of the Gaussian variables $\eta(j)$. Recall that
   $m_{\eta}^{(0)}$ and $S^{\eta}$ denote the mean and the covariance of
   $\eta(k)$ under hypothesis $H_0$. After conditioning on $W(1),W(2),...,W(k-1)$, using the independence of $\eta(j)$ and $\eta(s)$ $s \neq j$,
   and using the expression for the moment generating function of $\eta(j)$, we obtain successively:
\begin{eqnarray}
\mathcal{C}^\prime(\mu) &=&
 \mathbb E  \left[
{\mathbb{E}}_0 \left[    \mathrm{exp} \, \left(  \,\lambda^\top \sum_{j=1}^{k-1} {\Phi}(k, j) \eta(j)  \right)
     \right]
| W(1),...,W(k-1)
\right]             \nonumber\\
&=&
  \mathbb E  \left[{\mathbb E}_0  \left[
\Pi_{j=1}^{k-1} \mathrm{exp}  \left( \,\lambda^\top  \Phi(k,j) \eta(j)   \right)
\right]
| W(1),...,W(k-1)
\right]\nonumber\\
&=& \mathbb E \left[
\Pi_{j=1}^{k-1}
\mathbb{E} \left[
\mathrm{exp} \,  \left(   \,\lambda^\top \Phi(k,j)\, m_{\eta}^{(l)}   \right)
\mathrm{exp} \,  \left(  \frac{1}{2} \lambda^\top \Phi(k,j)^\top S^{\eta} \Phi(k,j) \lambda \right)
\right] \right]   \nonumber\\
&=&
\mathbb{E}\,\, [
\Pi_{j=1}^{k-1}\,\, \mathrm{exp}\, \left( \,\lambda^\top \left(\widetilde{\Phi}(k,j)+J\right)\, m_{\eta}^{(0)}\right) \nonumber \\
&\,&\mathrm{exp}\, \left( \frac{1}{2}  \left( J + \widetilde{\Phi}(k,j) \right)^\top
S^{\eta}
\left(  J + \widetilde{\Phi}(k,j) \right)      \right)
]
. \nonumber
\end{eqnarray}

Denote further:
\begin{eqnarray}
\label{eqn-delta(k)}
\delta(k\mu) &:=&
\mathbb{E}  \mathrm{[\,}
\mathrm{exp} \left( \lambda^\top \sum_{j=1}^{k-1} \widetilde{\Phi}(k,j) m_{\eta}^{(l)} \right)\,
\mathrm{exp} \left(  \frac{1}{2}
\lambda^\top \left( \sum_{j=1}^{k-1} \widetilde{\Phi}(k,j)^\top S^{\eta}
\widetilde{\Phi}(k,j)
 \right)  \lambda
\right)\,\\
&\,&\mathrm{exp} \left(\frac{1}{2} \lambda^\top J S^{\eta} \sum_{j=1}^{k-1}
\widetilde{\Phi}(k,j) \lambda \right)\,
\mathrm{exp} \left(\frac{1}{2} \lambda^\top
\sum_{j=1}^{k-1}
\widetilde{\Phi}(k,j)^\top S^{\eta} J \lambda \right) \,
\mathrm{]}, \nonumber\\
\overline{\mathcal{C}}(k\mu)&:=& \mathrm{exp}\left( (k-1) \left(\lambda^\top J m_{\eta}^{(l)} - \frac{1}{2} \lambda^\top J S^{\eta} J \lambda\right)\right),
\end{eqnarray}
where dependence on $H_0$ is dropped in the definition of $\delta(k\mu)$.
Then, it is easy to see that $\mathcal{C}^\prime(k \mu)= \overline{\mathcal{C}}(k\mu) \delta(k \mu)$.

Recall the expressions for $v$, $m_L^{(0)}$ and $\sigma_L^2$ in eqns. \eqref{eqn-m_L-sigma-L}, \eqref{eqn-def-v}.
After straightforward algebra, it can be shown that $\overline{\mathcal{C}}(k\mu)$ equals the expression in eqn.~\eqref{eqn-C-cal-over}.
Remak further that $\delta(k \mu)$ equals:
\begin{eqnarray}
\delta(k\mu) = \mathbb{E} \left[ \mathrm{exp}\left((\mu - \mu^2) e_i^\top
\sum_{j=1}^{k-1} \widetilde{\Phi}(k,j) m_{\eta}^{(0)} \right) \mathrm{exp}
\left(  \frac{\mu^2}{2} e_i^\top \sum_{j=1}^{k-1} \widetilde{\Phi}(k,j)S^{\eta}\widetilde{\Phi}(k,j)^\top \right) e_i\right].
\end{eqnarray}
Recall that $\overline{m}_0:=\max_{i=1,...,N}|[m_{\eta}^{(0)}]_i|$. Now, $\delta(k \mu)$ can be bounded from above as follows:

\begin{eqnarray}
\delta(k\mu) &\leq& \mathbb{E} \left[ \mathrm{exp}\left(|\mu - \mu^2|\,
\overline{m}_0 \sum_{j=1}^{k-1} \sum_{l=1}^{N} \,| [\widetilde{\Phi}(k,j)]_{il}|\right)\,
 \mathrm{exp}
\left(  \frac{\mu^2}{2} \|S^{\eta}\| \sum_{j=1}^{k-1} |[\widetilde{\Phi}(k,j)]_{il}]|^2 \right)\right]\\
&\leq&
\mathbb{E} \left[ \mathrm{exp}\left(|\mu - \mu^2|\,
\overline{m}_0 \sum_{j=1}^{k} \sum_{l=1}^{N} \,| [\widetilde{\Phi}(k,j)]_{il}|\right)\,
 \mathrm{exp}
\left(  \frac{\mu^2}{2} \|S^{\eta}\| \sum_{j=1}^{k} |[\widetilde{\Phi}(k,j)]_{il}]|^2 \right)\right].\nonumber
\end{eqnarray}

Using Lemmas \ref{lemma-event-a-l}, \ref{lemma-phi-epsilon} and \ref{lemma-stoch-matr},
using the total probability with the partition $\mathcal{A}_j$, $j=1,...,\mathcal{J}+1$, we get the inequality
 \eqref{eqn-append}.

\section{Summary}
%This paper studied iterative, distributed detection via running consensus, where
%no fusion node is required, and the sensors cooperate only with single-hop neighbors
% over a randomly varying network. We applied G{\"a}rtner-Ellis theorem, from large deviations,
% to show that, when the number of time steps $k$ grows,
% the Bayes error probability, at each sensor, decays as $e^{-k I_{(0)}(0)}$.
%  The quantity $I_{(0)}(0)$, the Chernoff in formation,
%  is (the best achievable) rate of the optimal centralized detector
%  that accesses all sensors' observations. This
%  is a promising result in distributed detection, as a random network with arbitrary connectivity,
%   connected only on average, can achieve asymptotic optimality.
%   Simulation examples indeed go along with theoretical findings. Even when
%    $k$ is a finite number, the best among all sensors is close,
%     in terms of the Bayes error probability, to the optimal centralized detector.`

%\newpage
\vspace{-4mm}
\bibliographystyle{IEEEtran}
\bibliography{IEEEabrv,LDPBibliography}

\begin{thebibliography}{10}
\providecommand{\url}[1]{#1}
\csname url@samestyle\endcsname
\providecommand{\newblock}{\relax}
\providecommand{\bibinfo}[2]{#2}
\providecommand{\BIBentrySTDinterwordspacing}{\spaceskip=0pt\relax}
\providecommand{\BIBentryALTinterwordstretchfactor}{4}
\providecommand{\BIBentryALTinterwordspacing}{\spaceskip=\fontdimen2\font plus
\BIBentryALTinterwordstretchfactor\fontdimen3\font minus
  \fontdimen4\font\relax}
\providecommand{\BIBforeignlanguage}[2]{{%
\expandafter\ifx\csname l@#1\endcsname\relax
\typeout{** WARNING: IEEEtran.bst: No hyphenation pattern has been}%
\typeout{** loaded for the language `#1'. Using the pattern for}%
\typeout{** the default language instead.}%
\else
\language=\csname l@#1\endcsname
\fi
#2}}
\providecommand{\BIBdecl}{\relax}
\BIBdecl

\bibitem{Varshney-I}
R.~Viswanatan and P.~R. Varshney, ``Decentralized detection with multiple
  sensors: Part {I}--fundamentals,'' \emph{Proc. IEEE}, vol.~85, pp. 54--63,
  January 1997.

\bibitem{Veraavali}
J.~F. Chamberland and V.~Veeravalli, ``Decentralized dectection in sensor
  networks,'' \emph{IEEE Transactions on Signal Processing}, vol.~51, no.~2,
  pp. 407--416, February 2003.

\bibitem{Tsitsiklis-detection}
J.~N. Tsitsiklis, ``Decentralized detection,'' \emph{Adv. Statis. Signal
  Processing}, vol.~2, pp. 297--344, 1993.

\bibitem{Poor-II}
R.~S. Blum, S.~A. Kassam, and H.~V. Poor, ``Decentralized detection with
  multiple sensors: Part {II}--advanced topics,'' \emph{Proc. IEEE}, vol.~85,
  pp. 64--79, January 1997.

\bibitem{Moura-detection-consensus}
S.~Kar, S.~A. Aldosari, and J.~M.~F. Moura, ``Topology for distributed
  inference on graphs,'' \emph{IEEE Transactions on Signal Processing},
  vol.~56, no.~6, pp. 2609–--2613, June 2008.

\bibitem{moura-cons-detection}
S.~Kar and J.~M.~F. Moura, ``Consensus based detection in sensor networks:
  topology design under practical constraints,'' in \emph{Proc. Workshop Inf.
  Theory in Sensor Networks}, Santa {F}e, NM, June 2007.

\bibitem{running-consensus-detection}
P.~Braca, S.~Marano, V.~Matta, and P.~Willet, ``Asymptotic optimality of
  running consensus in testing binary hypothesis,'' \emph{IEEE Transactions on
  Signal Processing}, vol.~58, no.~2, pp. 814--825, February 2010.

\bibitem{Sayed-detection-2}
F.~S. Cattivelli and A.~H. Sayed, ``Diffusion {LMS}-based detection over
  adaptive networks,'' in \emph{Proc. Asilomar Conf. Signals, Systems and
  Computers}, Pacific Grove, CA, October 2009.

\bibitem{Sayed-detection}
------, ``Distributed detection over adaptive networks based on diffusion
  estimation schemes,'' in \emph{Proc. IEEE SPAWC '09, 10th IEEE International
  Workshop on Signal Processing Advances in Wireless Communications}, Perugia,
  Italy, June 2009, pp. 61–--65.

\bibitem{weight-opt}
D.~Jakovetic, J.~Xavier, and J.~M.~F. Moura, ``Weight optimization for consenus
  algorithms with correlated switching topology,'' \emph{IEEE Transactions on
  Signal Processing}, vol.~58, no.~7, pp. 3788--3801, July 2010.

\bibitem{BoydGossip}
S.~Boyd, A.~Ghosh, B.~Prabhakar, and D.~Shah, ``Randomized gossip algorithms,''
  \emph{IEEE Transactions on Information Theory}, vol.~52, no.~6, pp.
  2508--2530, June 2006.

\bibitem{Sayed-LMS}
C.~G. Lopes and A.~H. Sayed, ``Diffusion least-mean squares over adaptive
  networks: formulation and performance analysis,'' \emph{IEEE Transactions on
  Signal Processing}, vol.~56, no.~7, pp. 3122–--3136, July 2008.

\bibitem{Sayed-LMS-new}
F.~S. Cattivelli and A.~H. Sayed, ``Diffusion {LMS} strategies for distributed
  estimation,'' \emph{IEEE Transactions on Signal Processing}, vol.~58, no.~3,
  pp. 1035–--1048, March 2010.

\bibitem{Giannakis-LMS}
I.~D. Schizas, G.~Mateos, and G.~B. Giannakis, ``Distributed {LMS} for
  consensus-based in-network adaptive processing,'' \emph{IEEE Trans. on Signal
  Processing}, vol.~57, no.~6, pp. 2365–--2381, June 2009.

\bibitem{Giannakis-LMS-2}
------, ``Distributed recursive least-squares for consensus-based in-network
  adaptive estimation,'' \emph{IEEE Trans. on Signal Processing}, vol.~57,
  no.~11, pp. 4583–--4588, November 2009.

\bibitem{SoummyaEst}
\BIBentryALTinterwordspacing
S.~Kar, J.~M.~F. Moura, and K.~Ramanan, ``Distributed parameter estimation in
  sensor networks: Nonlinear observation models and imperfect communication,''
  August 2008, submitted for publication, 51 pages. [Online]. Available:
  \url{arXiv:0809.0009v1 [cs.MA]}
\BIBentrySTDinterwordspacing

\bibitem{Kassam}
S.~A. {K}assam, \emph{Signal Detection in Non-Gaussian Noise}.\hskip 1em plus
  0.5em minus 0.4em\relax New York: Springer-Verlag, 1987.

\bibitem{allerton}
D.~Bajovic, D.~Jakovetic, J.~Xavier, B.~Sinopoli, and J.~M.~F. {M}oura,
  ``Distributed detection over time varying networks: large deviations
  analysis,'' in \emph{48th Allerton Conference on Communication, Control, and
  Computing}, Monticello, IL, Oct. 2010.

\bibitem{Cover}
T.~M. Cover and J.~A. Thomas, \emph{Elements of information theory}.\hskip 1em
  plus 0.5em minus 0.4em\relax New York: John Wiley and Sons, 1991.

\bibitem{Moura-saddle-point}
S.~A. Aldosari and J.~M.~F. Moura, ``Detection in sensor networks: the
  saddlepoint approximation,'' \emph{IEEE Transactions on Signal Processing},
  vol.~55, no.~1, pp. 327--340, January 2007.

\bibitem{DemboZeitouni}
A.~Dembo and O.~Zeitouni, \emph{Large deviations techniques and
  applications}.\hskip 1em plus 0.5em minus 0.4em\relax {B}oston, {MA}: Jones
  and Barlett, 1993.

\bibitem{Harry-van-Trees-book}
H.~L.~V. Trees, \emph{Detection, Estimation, and Modulation Theory - Part l -
  Detection, Estimation, and Linear Modulation Theory}.\hskip 1em plus 0.5em
  minus 0.4em\relax John Wiley and Sons, 2001.

\bibitem{running-consensus}
P.~Braca, S.~Marano, and V.~Matta, ``Enforcing consensus while monitoring the
  environment in wireless sensor networks,'' \emph{IEEE Transactions on Signal
  Processing}, vol.~56, no.~7, pp. 3375--3380, July 2008.

\end{thebibliography}
\end{document}